\documentclass[letterpaper, 10pt, onecolumn]{IEEEtran}
\usepackage{mathtools,dsfont,amsfonts}
\usepackage{enumitem}
\usepackage{subcaption}
\usepackage[mathscr]{euscript}

\usepackage[amsmath,thmmarks]{ntheorem}
\allowdisplaybreaks
\DeclareMathOperator*{\argmax}{arg\,max}

\usepackage{xcolor}

\newtheorem{problem}{Problem}
\newtheorem{assumption}{Assumption}

\theorembodyfont{\upshape}
\usepackage[svgnames]{xcolor}
\usepackage{hyperref}
\hypersetup{
    colorlinks=true,
    linkcolor=blue,
    filecolor=magenta,      
    urlcolor=cyan,
    citecolor=Green,
    }
    
\usepackage{booktabs}
\usepackage{subcaption}
\usepackage{graphicx}
\usepackage[sort,compress]{cite}

\usepackage{algorithm}%
\usepackage{algorithmicx}%
\usepackage{algpseudocode}%

\begin{document}
\title{Messaging Strategies for Incentivizing Agents in Dynamic Systems}
\author{Renyan Sun and Ashutosh Nayyar%
\thanks{Renyan Sun and Ashutosh Nayyar are with the Department of Electrical and Computer Engineering, University of Southern California, Los Angeles, CA, USA (email: renyansu@usc.edu, ashutosn@usc.edu). This work was supported in part under {NSF} grants ECCS 2025732 and ECCS 1750041.}.%
}
\maketitle

\begin{abstract}
We consider a finite-horizon discrete-time dynamic system jointly controlled by a designer and multiple agents, where the designer can influence the agents' actions through selective information disclosure. At each time step, the designer sends private messages to the agents from prespecified message spaces. The designer may also take an action that directly influences system dynamics and rewards. Each agent uses its received message and its own information to choose its action. We are interested in the setting where the designer would like to incentivize the agents to use prescribed strategies. We consider a notion of incentive compatibility that is based on sequential rationality at each realization of the common information among the designer and the agents. We consider both myopic and non-myopic agents and formulate the designer's problem as maximizing its total expected reward subject to the corresponding sequential rationality constraints. Under certain assumptions on the information structure, we obtain an equivalent reformulation of the designer's problem and develop a backward-inductive procedure based on a family of linear programs. For non-myopic agents, this procedure yields feasible solutions  but need not be globally optimal for a general designer reward. For myopic agents, the procedure yields a globally optimal solution. 
\end{abstract}


\section{Introduction}
\label{sec:introduction}
\IEEEPARstart{D}{ynamic} games model sequential decision-making problems where multiple self-interested agents take actions to influence the evolution of a dynamic system. The information structure of a dynamic game, i.e. a  specification of what information is available to each agent for each decision it has to make, plays a crucial role in the study of such games. Dynamic games with a variety of information structures have been investigated in the literature \cite{basar1999dynamic, filar2012competitive,maskin2001markov, gensbittel2015value, zheng2013decomposition, li2014lp, li2018lp, kartik2021upper, gupta2014common, gupta2016dynamic, ouyang2015dynamic,nayyar2013common, ouyang2016dynamic, tavafoghi2016stochastic,vasal2018systematic,tang2023dynamic}.  The most basic solution concept in these games is the Nash equilibrium -- a strategy profile (i.e. a strategy for each agent) where no agent has an incentive to deviate unilaterally \cite{osborne1994course}. Various refinements and modifications of Nash equilibrium such as Markov Perfect Equilibrium \cite{maskin2001markov}, sub-game perfect equilibrium \cite{osborne1994course}, common information based equilibria  \cite{nayyar2013common, ouyang2016dynamic, vasal2018systematic}, $\epsilon-$ approximate correlated equilibrium  \cite{pmlr-v202-liu23ay} have been developed and studied in the literature. However, because agents are strategic   and  only interested in optimizing their own individual objectives, the agents' behavior and the resulting outcomes that emerge at equilibria may not be desirable from a system designer or a social welfare perspective. 
 
 One way a designer can try to influence agents' behavior is by selectively revealing some private  information (that only the designer knows) to  the agents.  Such information can alter an agent's belief about the state of the world and the other agents, thereby influencing the actions it takes to optimize its objective. In game theory and economics literature, the problem of incentivizing strategic agents to take actions aligned with the system designer's objective through selective information disclosure is referred to as the ``Information Design" or ``Bayesian persuasion'' problem (see \cite{bergemann2019information} and references therein).

Much of the existing literature on information design has focused on {static} problems that involve  one-shot decisions with no temporal evolution of the state of the world or of information. Starting with the work in  \cite{kamenica2011bayesian}, a  variety of static information design problems have been investigated \cite{kamenica2019bayesian,bergemann2019information,akyol2016information} including those with multiple agents \cite{bergemann2016information, bergemann2016bayes, tavafoghi2017informational}, agents with different prior beliefs \cite{alonso2016bayesian}, multiple designers \cite{gentzkow2017bayesian, li2018multplesenders} and multi-dimensional state \cite{tamura2012theory}.  

However, many real-world scenarios involve dynamic environments where the state of the world and/or information about it evolves  over time, giving rise to \emph{dynamic information design} problems. In such settings, a designer can  disclose information sequentially over time, and agents may need to take a sequence of actions based on evolving information. A designer interested in long-term performance must consider the implications of information disclosure on both present and future agent behavior. Similarly, agents may need to take into account the effects of their current actions on future outcomes as well as future information. Such temporal interdependencies make dynamic information design problems particularly challenging. 

Dynamic information design problems differ in the temporal structure of the agents' and the designer's objectives. In some settings, agents and/or the designer are myopic because they are short-lived or because their objectives can be optimized stagewise. In the latter case, myopia does not imply suboptimal behavior, since stagewise optimization is optimal for the stated objective. Works where both the sender (designer) and the receiver (agent) of information are myopic include \cite{farokhi2016estimation} and the myopic formulation considered in \cite{sayin2016strategic}. In contrast, \cite{best2024persuasion, ely2017beeps, renault2017optimal, sayin2021deception, sayin2019optimality, sayin2019hierarchical} 
consider settings where a long-term-optimizing designer interacts with either a myopic agent or a sequence of short-lived agents. A common feature of these myopic formulations is that the agents' actions do not affect the evolution of the underlying system dynamics, or that the model does not involve a controlled state. In contrast, we consider a controlled dynamic system in which the actions of myopic agents can directly affect the subsequent system state.
Settings where both the designer and the agents seek to optimize their respective long-term objectives has also been explored in several works including \cite{farhadi2018static,meigs2020optimal, au2015dynamic, doval2020sequential, ely2020moving, tavafoghi2017informational, farhadi2022dynamic }. These works explore the challenges associated with balancing present and future incentives and offer insights into how strategic information disclosure can influence multi-stage decision-making processes. \cite{au2015dynamic, doval2020sequential, ely2020moving} assume that the state of the world  does not  change with time whereas   \cite{tavafoghi2017informational, farhadi2018static, farhadi2022dynamic, meigs2020optimal} model the state of the world as an uncontrolled Markov chain. In contrast, the model we consider allows for both the agents and the designer to take actions that influence the evolution of the system state.

We consider a dynamic setting with one designer and two agents. Both the designer and the agents may have some  information about the current state of the dynamic system. We partition the information at each time into common information (available to designer and the agents) and private information. At each time step, the designer sends private messages {from a prespecified message space} to the agents. The designer may also take an action that directly influences system dynamics and rewards. Each agent uses the received message (and its own information) to choose its action. We are interested in the setting where the designer would like to incentivize the two agents to play specific strategies $h^1,h^2$, respectively (see Section \ref{sec:model_two_agents} for details). We consider a notion of incentive compatibility that is based on sequential rationality at each realization of the  common information between the designer and the agent. This version of incentive compatibility, which we refer to as \emph{common information based sequential rationality},  is stricter than a  Nash equilibrium based notion of incentive compatibility that does not take the sequential nature of the problem into account. Our version of incentive compatibility requires the agents to be incentivized to use strategies $h^1,h^2$ at each time and for each realization of common information.
We consider both myopic and non-myopic agents. For myopic agents, this requirement is based on the immediate expected reward, whereas for non-myopic agents it is based on the expected reward-to-go.
These two notions of CISR lead to three related designer's problems. For non-myopic agents, we consider both an optimization problem, in which the designer maximizes its total expected reward subject to the CISR constraints, and the corresponding feasibility problem, in which the objective is only to find a designer strategy that satisfies these constraints. For myopic agents, we consider the analogous optimization problem with myopic CISR constraints. These three problems are summarized by a unified parameterized formulation in Section \ref{sec:model_two_agents}.

Our solution approach proceeds as follows. For an arbitrarily fixed designer strategy, we first derive \emph{necessary and sufficient} conditions under which the prescribed strategies \(h^1\) and \(h^2\) satisfy common information based sequential rationality. We then use these conditions to obtain an equivalent reformulation of the designer's problem. Motivated by the structure of this reformulation, we identify a sub-class of feasible solutions (referred to as sequentially maximizing tuples) that can be obtained by a backward-inductive procedure. This procedure relies on solving a family of linear programs - one for each time and each realization of common information. For non-myopic agents, this procedure yields feasible solutions and thus solves the corresponding  feasibility problem, although it need not, in general, yield a globally optimal solution for the designer. For myopic agents, the procedure yields a globally optimal solution. Under additional conditions, the computational burden can  be reduced to one linear program for each time and each distinct common-information based belief.

The models we consider allow for a variety of system dynamics and information structures. As noted earlier, both the designer and the  agents  may have some private information and the   ability to take actions that directly influence the evolution of the system state. In contrast to prior work, the designer in our setting wants to incentivize specific strategies using messages from prespecified message spaces. 

\emph{Notation:}
Random variables are denoted by upper case letters and their realizations by corresponding lower case letters. All random variables take values in finite sets which are denoted by the calligraphic font of the corresponding upper case letter. For time indices $t_1\leq t_2$, ${X}_{t_1:t_2}$ is a short hand notation for the collection of variables  $({X}_{t_1},{X}_{t_1+1},...,{X}_{t_2})$. Similarly, ${X}^{0:2}$ is a short hand notation for  $({X}^0, {X}^1, {X}^2)$ and $X^{i_1, \ldots, i_n} = (X^{i_1}, \ldots, X^{i_n})$ for some integers $i_1,\ldots,i_n$. $\mathbb{P}(\cdot)$ denotes the probability of an event; $\mathbb{E}[\cdot]$ denotes the expectation of a random variable. $\mathbb{P}^g(\cdot)$ (resp. $\mathbb{E}^g[\cdot]$) denotes that the probability (resp. expectation) depends on the choice of  function(s) $g$. The conditional probability $\mathbb{P}(\cdot|C_t=c_t)$ is sometimes written as  $\mathbb{P}(\cdot|c_t)$. 
$M^1_t \sim D_t^d$ means that $M^1_t$ is randomly generated according to the distribution $D^d_t$. For a set $A$, we denote its cardinality by $|A|$. 

\emph{Organization:}
The rest of the paper is organized as follows. We present the model and problem formulation in Section \ref{sec:model_two_agents}. In Section \ref{sec:constraints}, we derive a reformulation of the common information based sequential rationality constraints. In Section \ref{sec:global formulation}, we give an equivalent  reformulation of the designer's problem. In Section \ref{sec:seq max}, we define sequentially maximizing tuples, develop their backward-inductive construction, and provide our results for  non-myopic agents. In Section \ref{sec:optimality_guarantees}, we provide results for myopic agents and present the belief-state based computational reduction. We conclude in Section \ref{sec:conclusion}.

\section{Model and Problem Formulation}\label{sec:model_two_agents}
We consider a finite-horizon, discrete-time dynamic system that is jointly controlled by a designer and multiple agents. For simplicity, we consider a model with one designer and two agents (agent 1 and agent 2), but our approach naturally extends to $K>2$ agents. For each time $t\in\{1,2,\ldots,T\}$, $X_t \in \mathcal{X}_t$ is the state of the system at time $t$, $U_t^0 \in \mathcal{U}^0_t$ is the designer's action at time $t$, and $U_t^i \in \mathcal{U}^i_t$ is agent $i$'s action at time $t$, for $i\in\{1,2\}$. The system evolves as follows
\begin{equation}\label{eq:multidyn}
    X_{t+1} = f_t(X_t, U_t^0, U_t^1, U_t^2, N_t),
\end{equation}
where $N_t \in \mathcal{N}_t$ is the noise in the dynamic system at time $t$.

At each time $t$, the designer gets some private information about the state of the system. The designer can send messages (from prespecified message spaces) to the agents to influence their behavior. The designer's objective in sending these messages is to strategically reveal information to the agents in order to incentivize them to behave in a manner preferred by the designer. We will make this objective more precise later in this section. First, we describe the information structure and the messaging mechanism in more detail.

\emph{Information structure and messages:} Let $I_t^0$ denote the information available to the designer at time $t$. For each time $t$, we split $I_t^0$ into two components -- one is common (or public) information $C_t \in \mathcal{C}_t$ that is available to the designer as well as to both agents, the other is private information $P_t^0 \in \mathcal{P}_t^0$ that is available only to the designer. Similarly, let $I_t^i$ denote the information available to agent $i$ at the beginning of time $t$, for $i\in\{1,2\}$. This information can also be split into the common information $C_t$ and agent $i$'s private information $P_t^i \in \mathcal{P}_t^i$.

At each time $t$, the designer generates a \emph{message-action triplet} $(M_t^1, M_t^2, U_t^0) \in \mathcal{M}_t^1 \times \mathcal{M}_t^2 \times \mathcal{U}_t^0$. This triplet is generated randomly according to a probability distribution $D^d_t$. The distribution $D_t^d$ is selected by the designer as a function of its information at time $t$, i.e.,
\begin{equation}
  (M^1_t,M^2_t,U^0_t) \sim D_t^d,  \quad \text{and} \quad  D_t^d = g_t^d(P_t^0, C_t),
\end{equation}
where $g_t^d$ is referred to as the \emph{designer's strategy at time $t$}. We call the collection $g^d:=(g^d_1,g^d_2,\ldots,g^d_T)$ the designer's strategy. Let $\mathcal{G}^d$ denote the set of all possible strategies for the designer. 
With a slight abuse of notation, we will use $g^d_t(m^1_t, m^2_t, u_t^0 \mid p^0_t,c_t)$ to indicate
the \emph{probability} of generating the message-action triplet $m^1_t, m^2_t, u^0_t$ when the   designer is using the strategy $g^d_t$ at time $t$ and the realizations of its private and common information are $p^0_t,c_t$ respectively.

The designer
privately sends \(M_t^i\) to agent \(i\), for \(i\in\{1,2\}\).
Agent $i$ generates \(U_t^i\) as a function of its information and the message it receives:
\begin{equation}
\label{eq:multiagent_g}
    U_t^i=g_t^i(M_t^i,P_t^i,C_t),
    \qquad i\in\{1,2\},
\end{equation}
where $g_t^i$ is  agent $i$'s action strategy at time $t$. The collection $g^i := (g^i_1,g^i_2,\ldots,g^i_T)$ is referred to as \emph{agent $i$'s action strategy}. Let $\mathcal{G}^i$ denote the set of all possible action strategies for agent $i$. The strategy triplet $g:= (g^d, g^1, g^2)$ is called the \emph{strategy profile}. $(g^d, g^1, g^2)_{t:T}$ denotes the strategies used from time $t$ to $T$.

We assume that the initial state $X_1$ and the noise variables $N_t$, $t = 1,2,\ldots,T$, are finite-valued, mutually independent random variables with the distribution of $X_1$ being $P_{X_1}$ and the distribution of $N_t$ being $Q_t$. Further, all system variables (i.e., states, actions, messages, common and private information, etc.) take values in finite sets.

\emph{Reward structure:}
At each time \(t\), the designer receives a reward \(r_t^0(X_t,U_t^{0:2})\), and agent \(i\) receives a reward
\(r_t^i(X_t,U_t^{0:2})\) for $i\in\{1,2\}$. Although the messages do not enter the reward
functions directly, they influence the agents' actions and therefore
affect the realized rewards.

We make the following assumptions about the system model and the information structure.
\begin{assumption}\label{assm:1_two_agents}
    We assume that the private and common information evolve in the following manner.
    \begin{enumerate}
        \item Private information $P_t^0$ for the designer is given as: for $t=1,\ldots,T-1$
        \begin{equation}
            P^0_{t+1} = {\xi}^0_{t+1}(X_t, P_t^0, U_t^0, U_t^1, U_t^2, M_t^1, M_t^2, N_t),
        \end{equation}
        while private information $P_t^i$ for agent $i$ (where $i\in\{1,2\}$) is given as: for $t=1,\ldots,T-1$
        \begin{equation}
            P^i_{t+1} = {\xi}^i_{t+1}(X_t, P_t^i, U_t^0, U_t^1, U_t^2, M_t^i, N_t), 
        \end{equation}
        where ${\xi}^{0:2}_{t+1}$ are fixed functions.
        \item For $t=1,\ldots,T$, the common information at time $t+1$, $C_{t+1}$, consists of the common information at time $t$, $C_t$, and an increment $Z_{t+1}$. Further, $Z_{t+1}$ is given as
        \begin{equation}\label{equ: info increment2}
            Z_{t+1} = \zeta_{t+1}(X_t, P_t^{0:2}, U_t^{0:2}, M_t^{1,2}, N_t),
        \end{equation}
        where $\zeta_{t+1}$ is a fixed function.
        \item At time $t=1$, $P^{0}_1, P^1_1, P^2_1$ and $C_1$ are generated based on $X_1$ according to a given conditional distribution $\Lambda(p^0_1, p^1_1, p^2_1, c_1 \mid x_1)$.
    \end{enumerate}
\end{assumption}
Given a strategy profile $g = (g^d,g^1,g^2)$, we will be interested in the conditional distribution of the state and private information given the common information at time $t$, i.e., $\mathbb{P}^{g}(X_t = \cdot, P_t^{0:2} = \cdot \mid  C_t)$. We will refer to these distributions as the \emph{common information based conditional beliefs} under the strategy profile $g$. We make the following assumption about these beliefs.

\begin{assumption}[\textit{Strategy-independent Beliefs}] \label{assm:2}
The common information based conditional beliefs do not depend on the strategy profile. More precisely, consider any two strategy profiles $g=(g^d,g^1,g^2)$ and $\tilde{g}=(\tilde{g}^d,\tilde{g}^1,\tilde{g}^2)$ and any realization $c_t$ of common information $C_t$ that has non-zero probability under the two strategy profiles. Then, the corresponding common information based conditional beliefs are the same, i.e.,
\begin{align}\label{equ:strategy_independence_two_agents}
    \mathbb{P}^{g}(X_t = \cdot, P_t^{0:2} = \cdot \mid c_t)
    = \mathbb{P}^{\tilde{g}}(X_t = \cdot, P_t^{0:2} = \cdot \mid c_t).
\end{align}
\end{assumption}
Since the common information based conditional beliefs are strategy-independent under Assumption \ref{assm:2}, we can associate a unique belief with each realization of common information, i.e., given a realization $c_t$ of common information at time $t$, we can define the following belief on $X_t,P^{0:2}_t$:
\begin{equation}\label{def:new belief multi}
    \pi_t \left(x, p^{0:2} \mid c_t\right) := \mathbb{P}^{g}\left(X_t = x, P_t^{0:2} = p^{0:2} \mid C_t = c_t\right),
\end{equation}
where $g$ is any strategy profile under which $c_t$ has non-zero probability.

Assumption \ref{assm:2} does not require the probability of reaching a realization $c_t$ to be strategy-independent. It only requires the conditional belief on the current state and private information given $c_t$ to be strategy-independent. Two instances of information structures satisfying Assumption \ref{assm:2} are:
\begin{enumerate}
    \item \emph{One-step delayed information sharing:}
    Let
    $$
    P_t^0=X_t,\quad P_t^1=P_t^2=\emptyset, \quad C_t=(X_{1:t-1},U_{1:t-1}^{0:2},M_{1:t-1}^{1,2}).
    $$
    Conditioned on $C_t=c_t$, the previous state and action realizations are fixed, so the distribution of $X_t$ is determined by the system dynamics and the noise distribution independently of the probabilities with which the strategy profile generated the realized actions and messages.

    \item \emph{Global state with private local states:}
    Suppose the system state is
    $X_t=(S_t,L_t^1,L_t^2)$, where $S_t$ is a global state observed by the designer and both agents, and $L_t^i$ is a local state observed only by agent $i$, $i\in\{1,2\}$. The state dynamics are given as follows:  
     \begin{align}
    S_{t+1}
    &=f_t^S(S_t,U_t^{0:2},N_t^S), \notag\\
    L_{t+1}^i
    &=f_t^i(S_t,U_t^{0:2},N_t^i),
    \qquad i\in\{1,2\},
    \end{align}
    where $N_t^S,N_t^1,N_t^2$ are mutually independent noise variables. The initial state $X_1=(S_1,L_1^1,L_1^2)$ has the given distribution $P_{X_1}$. The information structure is given by 
    $$
    P_t^0=\emptyset,\quad P_t^1=L_t^1,\quad P_t^2=L_t^2,
    \quad
    C_t=(S_{1:t},U_{1:t-1}^{0:2}).
    $$
    In this case, too, it can be verified that the conditional distribution of current state and private information (i.e. $S_t,L^1_t,L^2_t$) conditioned on $C_t=c_t$ is determined by the system dynamics and the noise distribution with no dependence on the strategy profile.
\end{enumerate}

Assumptions \ref{assm:1_two_agents} and \ref{assm:2} are analogous to the system model and information structure assumptions made in \cite{nayyar2013common}. We refer the reader to \cite{nayyar2013common} for more examples of models where these assumptions are satisfied.

\emph{Designer's preferred strategies for agents:}
The designer would like to incentivize the agents to use specific strategies $h^1$ and $h^2$. For example, the designer may be interested in incentivizing \emph{obedience} of its messages by the agents \cite{bergemann2016information,bergemann2016bayes,doval2020sequential}, i.e., the designer would like to have $U_t^i = h_t^i(M_t^i,P_t^i,C_t) = M_t^i$. Note that obedience assumes that the messages sent by the designer take values in the corresponding action spaces $\mathcal{U}^1_t$ and $\mathcal{U}^2_t$. In this paper, we will consider the more general case where the message and action spaces may be different and the designer's preferred strategies $h^1$ and $h^2$ for the agents may not necessarily be the obedient strategies. For example, assuming $M_t^1, U_t^1, P_t^1, C_t$ are all binary valued, the designer may be interested in incentivizing the following non-obedient strategy for agent 1:
\begin{equation}
U^1_t = h^1_t(M^1_t,P^1_t,C_t) = \begin{cases}
P^1_t &\text{if $M^1_t=0$},\\
C_t &\text{if $M^1_t=1$}.
\end{cases}
\end{equation}

\emph{Incentive compatibility for the agents:}
We consider two kinds of agents: myopic and non-myopic. For myopic agents, a basic incentive compatibility requirement is that their expected immediate reward under $h^i$ is at least as large as that achievable under any other strategy. We adopt a stronger notion of incentive compatibility under which \(h^i\) remains optimal for agent \(i\) at every time step and for every realization of the \emph{common information}. We refer to this requirement as myopic common information based sequential rationality (myopic CISR). For non-myopic agents, we adopt a common information based sequential rationality requirement that considers the expected reward-to-go rather than only the immediate reward; we refer to it as non-myopic CISR. We formally define these requirements below.

\begin{definition}[Non-Myopic Agents]
\label{def: multi sequential rationality}
For \(i\in\{1,2\}\), we say that agent \(i\)'s action strategy \(h^i\)
satisfies non-myopic common information based sequential rationality (CISR) with
respect to the designer strategy \(g^d\) and the action strategy
\(h^{-i}\) of the other agent\footnote{We use \(-i\) to denote the
other agent.} if for each time \(t\)
and each possible realization \(c_t\) of the common information at time
\(t\),
\begin{align}
    &\mathbb{E}^{(g^d,h^i,h^{-i})_{t:T}}
    \left[
        \sum_{k=t}^T r_k^i(X_k,U_k^{0:2})
        \Bigm| c_t
    \right]
    \geq
    \mathbb{E}^{(g^d,g^i,h^{-i})_{t:T}}
    \left[
        \sum_{k=t}^T r_k^i(X_k,U_k^{0:2})
        \Bigm| c_t
    \right]
    \:
    \forall g^i\in\mathcal{G}^i.
    \label{eq:cisr multi}
\end{align}
\end{definition}

The expectation on the left-hand side of \eqref{eq:cisr multi} is
interpreted as follows. At the beginning of time $t$, given \(c_t\), we have an associated belief \(\pi_t\) on \(X_t,P_t^{0:2}\), given by \eqref{def:new belief multi}. With \(C_t=c_t\), \(X_t,P_t^{0:2}\) distributed according to \(\pi_t\), and  all subsequent actions, messages, states, and information variables generated using the strategies \((g^d,h^i,h^{-i})_{t:T}\) and the system model, the left-hand side of \eqref{eq:cisr multi} is the conditional expected reward-to-go for agent \(i\). The right-hand side of \eqref{eq:cisr multi} has the same interpretation except that agent \(i\)'s strategy \(h^i\) is replaced by an alternative strategy \(g^i\) while the designer's strategy and the other agent's strategy remain fixed.

\begin{definition}[Myopic Agents]
\label{def:cisr myopic}
For \(i\in\{1,2\}\), we say that agent \(i\)'s action strategy \(h^i\) satisfies myopic common information based sequential rationality (CISR) with respect to the designer strategy \(g^d\) and the action strategy \(h^{-i}\) of the other agent if for each time \(t\) and each possible realization \(c_t\) of the common information at time \(t\),
\begin{align}
    &\mathbb{E}^{(g^d,h^i,h^{-i})_{t:T}}
    \left[
        r_t^i(X_t,U_t^{0:2})
        \Bigm| c_t
    \right]
    \geq
    \mathbb{E}^{(g^d,g^i,h^{-i})_{t:T}}
    \left[
        r_t^i(X_t,U_t^{0:2})
        \Bigm| c_t
    \right]
    \qquad
    \forall g^i\in\mathcal{G}^i.
    \label{eq:cisr myopic}
\end{align}
\end{definition}

For convenience, we combine \eqref{eq:cisr multi} and
\eqref{eq:cisr myopic} into the following single equation:
\begin{align}
    &\mathbb{E}^{(g^d,h^i,h^{-i})_{t:T}}
    \left[
        \sum_{k=t}^T \gamma^{k-t}
        r_k^i(X_k,U_k^{0:2})
        \Bigm| c_t
    \right]
    \geq
    \mathbb{E}^{(g^d,g^i,h^{-i})_{t:T}}
    \left[
        \sum_{k=t}^T \gamma^{k-t}
        r_k^i(X_k,U_k^{0:2})
        \Bigm| c_t
    \right]
    \qquad
    \forall g^i\in\mathcal{G}^i,
    \label{eq:cisr}
\end{align}
where \(\gamma\in\{0,1\}\), with the convention \(0^0=1\). \eqref{eq:cisr} is  a unified representation of the two CISR conditions: setting \(\gamma=1\) yields the non-myopic CISR condition in \eqref{eq:cisr multi}, whereas setting \(\gamma=0\) yields the myopic CISR condition in \eqref{eq:cisr myopic}.

If all the inequalities in Definition~\ref{def: multi sequential rationality} hold, we say that ``\(h^i\) satisfies non-myopic \(CISR(g^d,h^{-i})\)''; if they hold for both agents, we say that ``\((h^1,h^2)\) satisfies non-myopic \(CISR(g^d)\).'' Similarly, if all the inequalities in Definition~\ref{def:cisr myopic} hold, we say that ``\(h^i\) satisfies myopic \(CISR(g^d,h^{-i})\),'' and if they hold for both agents, we say that ``\((h^1,h^2)\) satisfies myopic \(CISR(g^d)\).''

The designer's total expected reward under a designer strategy $g^d$ when the agents use the prescribed strategies \(h^1\) and \(h^2\) is given by
\begin{equation}
\label{eq:multiagent designer reward-to-go}
    J^0(g^d,h^1,h^2)
    :=
    \mathbb{E}^{(g^d,h^1,h^2)}
    \left[
        \sum_{t=1}^T
        r_t^0(X_t,U_t^0,U_t^1,U_t^2)
    \right].
\end{equation}

We first consider the designer's optimization problem when the agents are non-myopic.

\begin{problem}[Optimization with Non-Myopic Agents]
\label{prob: multi}
Given fixed strategies \(h^1,h^2\), the designer's goal is to find an optimal strategy \(g^d\) that maximizes the designer's total expected reward while ensuring that \((h^1,h^2)\) satisfies non-myopic common information based sequential rationality with respect to \(g^d\), as specified in Definition~\ref{def: multi sequential rationality}. That is, the designer would like to solve the following strategy optimization problem:
\begin{align*}
    \max_{g^d\in\mathcal{G}^d}\quad
    &J^0(g^d,h^1,h^2)\\
    \text{\rm s.t.}\quad
    &(h^1,h^2)\text{ satisfies non-myopic }CISR(g^d).
\end{align*}
\end{problem}

A related problem arises when the designer is only interested in incentivizing agents to use $h^1,h^2$, and it does not have a reward objective of its own  (i.e., \(r_t^0\equiv 0\)). In this case,
Problem~\ref{prob: multi} reduces to the following feasibility problem.
\begin{problem}[Feasibility with Non-Myopic Agents]
\label{prob:feasibility}
Given fixed strategies \(h^1,h^2\), the designer's goal is to find a strategy \(g^d\) that ensures that \((h^1,h^2)\) satisfies non-myopic common information based sequential rationality with respect to \(g^d\), as specified in Definition~\ref{def: multi sequential rationality}. That is, the designer would like to solve the following feasibility problem:
\begin{align*}
    \max_{g^d\in\mathcal{G}^d}\quad
    &0\\
    \text{\rm s.t.}\quad
    &(h^1,h^2)\text{ satisfies non-myopic }CISR(g^d).
\end{align*}
\end{problem}

We also consider the designer's problem with myopic agents.
\begin{problem}[Optimization with Myopic Agents]
\label{prob:myopic agents}
Given fixed strategies \(h^1,h^2\), the designer's goal is to find an optimal strategy \(g^d\) that maximizes the designer's total expected reward while ensuring that \((h^1,h^2)\) satisfies myopic common information based sequential rationality with respect to \(g^d\), as specified in Definition~\ref{def:cisr myopic}. That is, the designer would like to solve the following strategy optimization problem:
\begin{align*}
    \max_{g^d\in\mathcal{G}^d}\quad
    &J^0(g^d,h^1,h^2)\\
    \text{\rm s.t.}\quad
    &(h^1,h^2)\text{ satisfies myopic }CISR(g^d).
\end{align*}
\end{problem}
\textbf{A Unified Problem Formulation:} For a fixed $\gamma \in \{0,1\}$ and designer rewards $r_{1:T}^0$, we define the parameterized problem:
\begin{align*}
    \mathbf{P}(\gamma, r^0_{1:T}): \quad\max_{g^d\in\mathcal{G}^d}\quad
    &J^0(g^d,h^1,h^2)\\
    \text{\rm s.t.}\quad
    &(h^1,h^2)\text{ satisfies  \eqref{eq:cisr} for all $t, c_t$.}
\end{align*}
The three problems considered above are instances of this unified formulation. Problem~\ref{prob: multi} is $\mathbf{P}(1, r_{1:T}^0)$, Problem~\ref{prob:feasibility} is the zero-objective feasibility problem $\mathbf{P}(1, r_{1:T}^0\equiv 0)$, whereas Problem~\ref{prob:myopic agents} is $\mathbf{P}(0, r_{1:T}^0)$.

\section{Reformulation of the CISR Constraints}\label{sec:constraints}
In this section, we consider a fixed \(\gamma\in\{0,1\}\) and work with the unified CISR representation in \eqref{eq:cisr}. Our goal is to develop a reformulation of the constraints in the unified problem $\mathbf{P}(\gamma, r^0_{1:T})$.
A key feature of the CISR conditions in \eqref{eq:cisr} is that they must be satisfied for each agent, at each time, and for each possible realization of the
common information. The first step in our solution approach is to derive a backward-inductive characterization of these constraints.

Recall that under Assumption \ref{assm:2} we can associate a  belief $\pi_t(\cdot, \cdot \mid c_t)$ with each realization $c_t$ of $C_t$ (see \eqref{def:new belief multi}). This belief is the conditional  distribution of $X_t, P^{0:2}_t$ given $C_t=c_t$ under any strategy profile where the realization $c_t$ can occur.  Using this belief and the designer strategy $g^d$, we  define a probability distribution on $X_t, P^{0:2}_t, M_t^{1,2}, U_t^0, N_t$ as follows.

\begin{definition}\label{def:eta multi}
    Given a designer  strategy $g^d$ and a realization $c_t$ of the common information at time $t$, we define the following common information based belief on $X_t,P^{0:2}_t, M^{1,2}_t,U^0_t,N_t$: 
    \begin{align}\label{eq:eta multi}
        &\eta_t(x_t, p^{0:2}_t,m_t^{1,2}, u_t^0, n_t
        \mid c_t) =Q_t(n_t)\pi_t(x_t,p^{0:2}_t \mid c_t)g^d_t(m_t^{1,2}, u_t^0\mid p^0_t,c_t),
    \end{align}
for all $x_t, p^{0:2}_t,m_t^{1,2}, u_t^0, n_t$. Here, $Q_t(\cdot)$ is the a priori probability distribution of noise $N_t$ and $\pi_t(\cdot, \cdot \mid c_t)$ is the strategy-independent common information based belief on $X_t,P^{0:2}_t$ given $c_t$.  Further, we denote by $\eta_t(p^i_t, m_t^i \mid c_t)$ the  probability of the event $\{P^i_t =p^i_t,M_t^i=m^i_t\}$  under the distribution $\eta_t(\cdot\mid c_t)$ 
\big($\eta_t(p^i_t, m_t^i \mid c_t)$ can be obtained by summing over all other arguments of $\eta_t(\cdot\mid c_t)$\big).
\end{definition}
It is straightforward to verify that
\begin{align}
    \mathbb{P}&^{(g^d,g^1, g^2)}\big(X_t = x_t, P_t^{0:2} = p^{0:2}_t, M^{1,2}_t=m^{1,2}_t, U_t^0 = u_t^0, N_t=n_t \mid  c_t\big)= \eta_t(x_t, p^{0:2}_t,m_t^{1,2}, u_t^0, n_t \mid  c_t),
\end{align}
for any designer strategy $g^d$, for any agents' action strategies $g^1,g^2$, and any $c_t$ which has a non-zero probability under the strategy profile $(g^d,g^1,g^2)$.

Consider a fixed designer strategy $g^d$. This $g^d$ induces $\eta_t$ as per Definition \ref{def:eta multi}. In this section, we will develop a reformulation of the condition \eqref{eq:cisr} in terms of linear inequalities involving $g^d$ and $\eta_t$. To that end, we define the following common information based value functions for agent
\(i\in\{1,2\}\) for $t \leq T$:
\begin{align}
    W_t^i(c_t)
    &:=
    \mathbb{E}^{(g^d,h^i,h^{-i})_{t:T}}
    \left[
        \sum_{k=t}^T
        \gamma^{k-t}r_k^i(X_k,U_k^{0:2})
        \Bigm|c_t
    \right].
    \label{eq:W-prescribed-payoff}
\end{align}
Note that \(W_t^i(c_t)\) is the left-hand side of
\eqref{eq:cisr}. We further define
\begin{equation}
    W_{T+1}^i(c_{T+1}) := 0.
    \label{eq:W_defa}
\end{equation}

By separating the immediate and future rewards in \eqref{eq:W-prescribed-payoff} and conditioning the future rewards on $C_{t+1}$, the value
function in \eqref{eq:W-prescribed-payoff} can be written recursively
as
\begin{align}
    W_t^i(c_t) &:= \mathbb{E}^{\eta_t(\cdot\mid c_t)}[r_t^i(X_t, U_t^0, h_t^1(M_t^1, P_t^1, c_t), h_t^2(M_t^2, P_t^2, c_t))+\gamma W_{t+1}^i(c_t, Z_{t+1})], \label{eq:W_def}
\end{align}
where $Z_{t+1}$ in \eqref{eq:W_def} is the common information increment at time $t+1$ defined  according to \eqref{equ: info increment2} with control actions   $U_t^i = h_t^i(M_t^i, P^i_t, c_t)$ for $i\in \{1,2\}$ and the expectation in \eqref{eq:W_def} is with respect to the probability distribution $\eta_t(\cdot\mid c_t)$ defined in Definition~\ref{def:eta multi}. 
More explicitly, for each $i\in \{1,2\}$, $W_t^i(c_t)$ can be written as the following expression:
\begin{align}\label{eq:multi agent value function time any t}
    &W_t^i(c_t) := \sum_{\substack{x\in\mathcal{X}_t,p^0\in\mathcal{P}_t^0,p^1\in\mathcal{P}_t^1\\p^2\in\mathcal{P}_t^2, m^1 \in \mathcal{M}_t^1,m^2 \in \mathcal{M}_t^2 \\ u_t^0\in \mathcal{U}_t^0, n\in \mathcal{N}_t}}
\eta_t(x,p^{0:2}, m^{1,2}, u_t^0, n \mid  c_t) \times \Big[r^i_t(x,u_t^{0:2})+ \gamma W_{t+1}^i(c_t, z_{t+1})\Big],
\end{align}
where $u_t^i = h_t^i(m^i, p^i, c_t), \:i\in \{1,2\},$ and  $z_{t+1} = \zeta_{t+1}(x,\allowbreak p^{0:2},\allowbreak u_t^{0:2}, m^{1,2}, n)$.

Recall that $\eta_t(\cdot \mid c_t)$ as defined in \eqref{eq:eta multi} is a probability distribution on $X_t, P^{0:2}_t, M_t^{1,2}, U_t^0, N_t$ and that $\eta_t(p_t^i, m_t^i \mid c_t)$ denotes the probability of the event $\{P^i_t =p^i_t,M_t^i=m^i_t\}$  under $\eta_t(\cdot \mid c_t)$. 
Consider a $p_t^i \in \mathcal{P}_t^i, m_t^i \in \mathcal{M}_t^i, c_t \in \mathcal{C}_t $ such that $\eta_t(p_t^i, m_t^i \mid c_t) > 0$.
Using $\eta_t(\cdot \mid c_t)$, we define the following distribution 
\begin{align}\label{eq:eta_conditioned multi}
    &\mu^i_t(x_t,p^{0}_t, p^{-i}_t,  m^{-i}_t, u_t^0, n_t\mid m^i_t,p^i_t,c_t) = \frac{\eta_t(x_t,p^{0}_t,p_t^i, p^{-i}_t, m_t^i, m_t^{-i}, u_t^0, n_t \mid  c_t)}{\eta_t(p_t^i, m_t^i \mid  c_t)}.
\end{align}

The interpretation of $\mu^i_t(\cdot \mid m^i_t,p^i_t,c_t)$ is that it is the conditional distribution of $X_t, P^{0,-i}_t, M_t^{-i}, U_t^0, N_t$ given $M_t^i=m_t^i,\allowbreak P_t^i = p_t^i$ when the joint distribution of $X_t, P^{0:2}_t, M_t^{1,2}, U_t^0, N_t$ is $\eta_t(\cdot\mid c_t)$. 

The following lemma provides a necessary and sufficient condition for \eqref{eq:cisr}.
\begin{lemma}\label{lem:suff nece multi}
    \eqref{eq:cisr} is satisfied if and only if the following statement is true:
    
    For $i \in \{1,2\}$, $t=T,T-1,\ldots 1,$ and for each  $c_t \in \mathcal{C}_t, m_t^i \in \mathcal{M}_t^i,p_t^i \in \mathcal{P}_t^i $ such that $\eta_t(p_t^i, m_t^i\mid  c_t) > 0$:
    \begin{align}
         &h_t^i(m_t^i, p_t^i, c_t) \in \argmax_{u \in \mathcal{U}_t^i} \mathbb{E}^{\mu^i_t(\cdot\mid m^i_t,p^i_t,c_t)}[r_t^i(X_t, U_t^0, h_t^j(M_t^j,P_t^j, c_t), u)+\gamma W^i_{t+1}(c_t, Z_{t+1})], \label{eq:suff1 multi}
        \end{align}
        where $j=- i$ and 
        $Z_{t+1}$ in \eqref{eq:suff1 multi} is the common information increment at time $t+1$ defined  according to (\ref{equ: info increment2}) with control actions  $U_t^j = h_t^j(M_t^j, P_t^j, c_t), U_t^i = u$, the message $M_t^i = m_t^i$ and the expectation is with respect to the distribution $\mu^i_t(\cdot\mid m^i_t,p^i_t,c_t)$ defined in ~\eqref{eq:eta_conditioned multi}.
\end{lemma}
\begin{proof}
    See Appendix \ref{app:multi equiv}.
\end{proof}

Lemma \ref{lem:suff nece multi} provides an alternative characterization of \eqref{eq:cisr}. 
This alternative characterization can be expressed as inequalities that are linear in $\eta_t$ (which itself is linear in $g^d_t$, see Definition \ref{def:eta multi}).
To do so, we first note that  \eqref{eq:suff1 multi}  can be written as the following collection of inequalities:
\begin{align}
&\mathbb{E}^{\mu^i_t(\cdot\mid m^i_t,p^i_t,c_t)}[r_t^i(X_t, U_t^0, U_t^j, u_t^i)+ \gamma W^i_{t+1}(c_t, Z_{t+1})]  \geq \mathbb{E}^{\mu^i_t(\cdot\mid m^i_t,p^i_t,c_t)}[r_t^i(X_t, U_t^0, U_t^j, u)+ \gamma W^i_{t+1}(c_t, \tilde Z_{t+1})],  ~\forall {u}\in \mathcal{U}^i_t, \label{eq:lp1}
\end{align}
     where $U_t^j = h_t^j(M_t^j, P_t^j, c_t)$, $u_t^i = h_t^i(m_t^i, p_t^i, c_t)$,  $Z_{t+1}$ on the left-hand side of the inequality above  is given by (\ref{equ: info increment2}) with $U_t^i = u_t^i, M_t^i = m_t^i$,  while  $\tilde{Z}_{t+1}$ on the right-hand side of the inequality is given by (\ref{equ: info increment2}) with $U_t^i = u, M_t^i = m_t^i$.
    
    Consider the right-hand side of \eqref{eq:lp1}. We can evaluate this expectation as follows:
    \begin{align}
    &\sum_{\substack{x\in\mathcal{X}_t,p^0\in\mathcal{P}_t^{0},p^j\in\mathcal{P}_t^{j} \\ m^j \in \mathcal{M}_t^j, u_t^0\in\mathcal{U}_t^0, n\in \mathcal{N}_t}}
        \mu^i_t(x,p^{0}, p^j, m^j, u_t^0, n\mid m^i_t,p^i_t,c_t)\times\Big[r^i_t(x,u_t^0, u_t^j, u)+ \gamma W_{t+1}^i(c_t,\tilde {z}_{t+1})\Big] \notag \\
    &=\sum_{\substack{x\in\mathcal{X}_t,p^0\in\mathcal{P}_t^{0},p^j\in\mathcal{P}_t^{j} \\ m^j \in \mathcal{M}_t^j, u_t^0\in\mathcal{U}_t^0, n\in \mathcal{N}_t}}
        \frac{\eta_t(x,p^{0},p_t^i, p^j, m_t^i, m^j, u_t^0, n \mid  c_t)}{\eta_t(p_t^i, m_t^i \mid  c_t)}\times\Big[r^i_t(x,u_t^0, u_t^j, u)+ \gamma W_{t+1}^i(c_t,\tilde{z}_{t+1})\Big], \label{eq:lp2}
    \end{align}
    where $u_t^j = h_t^j(m^j, p^j, c_t), \:\allowbreak \tilde z_{t+1} = \zeta_{t+1}(x,  p^{0,j},\allowbreak p^i_t,\allowbreak u_t^{0,j},\allowbreak u,\allowbreak m^j, m_t^i, n)$, and we have used the definition of $\mu_t^i$ from \eqref{eq:eta_conditioned multi}. (Recall that $\eta_t(p_t^i, m_t^i \mid  c_t) >0$ in \eqref{eq:suff1 multi}). Writing a similar expression for the left-hand side of \eqref{eq:lp1} and canceling $\eta_t(p_t^i, m_t^i \mid  c_t)$ results in the following set of inequalities that are linear in $\eta_t$:
\begin{align}
    &\sum_{\substack{x\in\mathcal{X}_t,p^0\in\mathcal{P}_t^{0},p^j\in\mathcal{P}_t^{j} \\ m^j \in \mathcal{M}_t^j, u_t^0\in\mathcal{U}_t^0, n\in \mathcal{N}_t}}
    \eta_t(x,p^{0}, p^i_t, p^j, m_t^i, m^j, u_t^0, n \mid  c_t)\times\Big[r^i_t(x,u_t^0, u_t^j, u_t^i)+ \gamma W_{t+1}^i(c_t,{z}_{t+1})\Big] \notag \\
    & \geq \sum_{\substack{x\in\mathcal{X}_t,p^0\in\mathcal{P}_t^{0},p^j\in\mathcal{P}_t^{j} \\ m^j \in \mathcal{M}_t^j, u_t^0\in\mathcal{U}_t^0, n\in \mathcal{N}_t}}
    \eta_t(x,p^{0}, p^i_t, p^j, m_t^i, m^j, u_t^0, n \mid  c_t) \times\Big[r^i_t(x,u_t^0, u_t^j, u)+ \gamma W_{t+1}^i(c_t,\tilde{z}_{t+1})\Big],~\forall {u}\in \mathcal{U}_t^i, \label{eq:multilp}
\end{align}
where $u_t^i = h_t^i(m_t^i, p_t^i, c_t), \:\allowbreak u_t^j = h_t^j(m^j, p^j, c_t), \allowbreak z_{t+1} = \zeta_{t+1}(x,  p^{0,j}, p^i_t, u_t^{0,j, i}, m^j, m_t^i, n),\:\allowbreak \tilde{z}_{t+1} = \zeta_{t+1}(x, \allowbreak p^{0,j}, p^i_t, u_t^0, u_t^j, u,  m^j, m_t^i, n)$. 

Thus, the condition in Lemma \ref{lem:suff nece multi} can be stated as follows: for all $p_t^i, m_t^i, c_t$ for which $\eta_t(p_t^i, m_t^i\mid  c_t) > 0$ the inequalities in \eqref{eq:multilp} hold. Further, if $\eta_t(p_t^i, m_t^i\mid  c_t) = 0,$  then it follows that $\eta_t(x,p^{0}, p^i_t, p^j, m_t^i, m^j, u_t^0, n \mid  c_t) =0$ for all $x,p^{0}, p^j,\allowbreak  m^j, u_t^0, n$ and hence  \eqref{eq:multilp} is trivially true since  both sides of the inequality are $0$. 
We  summarize the above discussion in the following theorem.

\begin{theorem}\label{thm:necce and suff multi}
    Consider an arbitrary designer strategy $g^d$. Define $\eta_t$ as in Definition~\ref{def:eta multi} and the common information based value functions $W_{T+1}^1,\ldots,W_1^1,W_{T+1}^2,\ldots,W_1^2$ using \eqref{eq:W_defa} and \eqref{eq:multi agent value function time any t}. Then, \eqref{eq:cisr} is satisfied if and only if the inequalities \eqref{eq:multilp} hold for each $i\in\{1,2\}$, $t=1,2,\ldots,T$, and for all $c_t,m_t^i,p_t^i$.
\end{theorem}
\begin{proof}
   From Lemma \ref{lem:suff nece multi}, we know that \eqref{eq:cisr} is satisfied if and only if   \eqref{eq:suff1 multi} holds for all  $p_t^i, m_t^i, c_t$ for which $\eta_t(p_t^i, m_t^i\mid  c_t) > 0$.  
  As discussed above, when $\eta_t(p_t^i, m_t^i\mid  c_t) > 0$,  \eqref{eq:suff1 multi} is equivalent to the  inequalities in \eqref{eq:multilp}.  Moreover, \eqref{eq:multilp} is trivially true if $\eta_t(p_t^i, m_t^i\mid  c_t) = 0$.  This proves the theorem.   
\end{proof}

\section{Global Reformulation}\label{sec:global formulation}
Theorem~\ref{thm:necce and suff multi} in the preceding section gives a recursive reformulation of the CISR constraints in terms of \(\eta_t\) and the agents' value functions
\(W_t^1,W_t^2\). To obtain a corresponding recursive reformulation of the designer's objective, we define the designer's common information based value functions as
\begin{align}
    V_t(c_t)
    =
    \mathbb{E}^{(g^d,h^1,h^2)_{t:T}}
    \left[
        \sum_{k=t}^T
        r_k^0(X_k,U_k^{0:2})
        \Bigm|c_t
    \right].
    \label{eq:V payoff identity}
\end{align}
Consequently,
\begin{align}
    J^0(g^d,h^1,h^2)
    =
    \sum_{c_1\in\mathcal C_1}
    \mathbb P(C_1=c_1)V_1(c_1).
    \label{eq:designer reward reformulation multi}
\end{align}
The value functions in \eqref{eq:V payoff identity} can be written recursively as
 \begin{equation}
    V_{T+1}(c_{T+1}) := 0, \label{eq:V_defa}
\end{equation}
and for $t \le T,$
\begin{align}
    V_t(c_t) &:= \mathbb{E}^{\eta_t(\cdot\mid c_t)}[r_t^0(X_t, U_t^0, h_t^1(M_t^1, P_t^1, c_t), h_t^2(M_t^2, P_t^2, c_t))+ V_{t+1}(c_t, Z_{t+1})], \label{eq:V_def designer}
\end{align}
where $Z_{t+1}$ in \eqref{eq:V_def designer} is the common information increment defined in (\ref{equ: info increment2}) with control actions  $U_t^i = h_t^i(M_t^i, P^i_t, c_t)$ for $i\in\{1,2\}$, and the expectation in \eqref{eq:V_def designer} is with respect to the probability distribution $\eta_t(\cdot\mid c_t)$ from Definition~\ref{def:eta multi}.
More explicitly, $V_t(c_t)$ can be written as the following expression:
\begin{align}\label{eq:designer value function time any t multi}
    &V_t(c_t) := \sum_{\substack{x\in\mathcal{X}_t,p^0\in\mathcal{P}_t^0,p^1\in\mathcal{P}_t^1\\p^2\in\mathcal{P}_t^2, m^1 \in \mathcal{M}_t^1,m^2 \in \mathcal{M}_t^2 \\ u_t^0\in \mathcal{U}_t^0, n\in \mathcal{N}_t}}
\eta_t(x,p^{0:2}, m^{1,2}, u_t^0, n\mid  c_t)\times\Big[r^0_t(x,u_t^{0:2})+ V_{t+1}(c_t, z_{t+1})\Big],
\end{align}
where $u_t^i = h_t^i(m^i, p^i, c_t)$ for $i \in \{1,2\}$, and  $z_{t+1} = \zeta_{t+1}(x, \allowbreak p^{0:2}, u_t^{0:2}, m^{1,2}, n)$. 
\begin{theorem}[Equivalence of the Global Problem]
\label{thm:GP equivalence}
Fix \(\gamma\in\{0,1\}\) and the designer rewards \(r_{1:T}^0\). The following optimization problem, referred to as  Global Problem \(\mathbf{GP}\), is an equivalent  reformulation of the parameterized problem \(\mathbf{P}(\gamma,r_{1:T}^0)\). Specifically, a designer strategy $g^d$ is feasible for \(\mathbf{P}(\gamma,r_{1:T}^0)\) if and only if $g^d$ is the designer-strategy component of a feasible solution of \(\mathbf{GP}\). Moreover, the two problems assign the same objective value to every such $g^d$.
\begin{align*}
    \mathbf{GP}:
                &\max_{g^d,\eta_{1:T},W^{1,2}_{1:T}, V_{1:T}}\sum_{c_1 \in \mathcal{C}_1}\mathbb{P}(C_1 = c_1) V_1(c_1) \\
                &\hspace{30pt}\text{\rm s.t.}\quad\text{for all } t=1:T\text{ and } c_t \in \mathcal C_t,\\
                &\hspace{57pt}\eta_t(\cdot\mid c_t), g^d_t(\cdot\mid \cdot, c_t)\text{ satisfy } \eqref{eq:eta multi}\:\forall x_t, p_t^{0:2}, m_t^{1,2}, u_t^0, n_t,\\
                &\hspace{57pt}\eta_t(\cdot \mid  c_t), W_{t+1}^i(c_t, \cdot) \text{ satisfy \eqref{eq:multilp} $(W_{T+1}^i(\cdot) = 0)$}\:\forall m_{t}^i, p_{t}^i, \text{ for each $i \in \{1,2\}$},\\
                &\hspace{57pt}\eta_t(\cdot \mid  c_t), W^i_{t}(c_t), W^i_{t+1}(c_t, \cdot)\text{ satisfy \eqref{eq:multi agent value function time any t} $(W_{T+1}^i(\cdot) = 0)$},\text{ for each $i \in \{1,2\}$},\\
                &\hspace{57pt}\eta_t(\cdot \mid  c_t), V_{t}(c_t), V_{t+1}(c_t, \cdot)\text{ satisfy \eqref{eq:designer value function time any t multi} $(V_{T+1}(\cdot) = 0)$}.
\end{align*}
\end{theorem}
\begin{proof} 
By Theorem~\ref{thm:necce and suff multi}, a designer strategy \(g^d\) is feasible in \(\mathbf{P}(\gamma,r_{1:T}^0)\) if and only if its associated beliefs \(\eta_{1:T}\) and agents' value functions \(W_{1:T}^1,W_{1:T}^2\) satisfy the constraints \eqref{eq:eta multi}, \eqref{eq:multilp}, and \eqref{eq:multi agent value function time any t}. Moreover, \eqref{eq:designer value function time any t multi} defines the associated designer's value functions, and \eqref{eq:designer reward reformulation multi} shows that the objective of \(\mathbf{GP}\) is equal to \(J^0(g^d,h^1,h^2)\). Hence, \(\mathbf{GP}\) and \(\mathbf{P}(\gamma,r_{1:T}^0)\) have the same feasible designer strategies and the same objective value for each feasible strategy.
\end{proof}

We note that the CISR conditions of \eqref{eq:cisr} need to hold for each alternative strategy  of agent $i$. Considering the condition at time $t$, the number of alternative choices for $g^i_{t:T}$ is exponentially large. (To be precise, the number of possible choices of $g^i_{t}$ is $exp(|\mathcal{M}^i_t \times \mathcal{P}^i_t \times \mathcal{C}_t|log|\mathcal{U}^i_t|)$.) Thus, \(\mathbf{P}(\gamma,r_{1:T}^0)\) involves an exponentially large number of inequality constraints. On the other hand,   the number of constraints in $\mathbf{GP}$ grows linearly with the size of  $\mathcal{C}_t, \mathcal{P}^{0:2}_t, \mathcal{M}^{1:2}_t$ (and the space of state, action, noise variables). Thus, $\mathbf{GP}$ is a more tractable optimization problem.
Although it is an exact reformulation of \(\mathbf{P}(\gamma,r_{1:T}^0)\), it may still be computationally difficult because its variables and constraints are coupled across time. In particular, it contains products of $\eta_t(\cdot\mid c_t)$ with $W_{t+1}^i(c_t,\cdot)$ and $V_{t+1}(c_t,\cdot)$, all of which are optimization variables. Hence, the constraints in $\mathbf{GP}$ are not jointly linear in the optimization variables. However, when the value functions at time $t+1$ are fixed, the constraints at time $t$ are linear in the remaining variables. This observation motivates the backward construction of sequentially maximizing solutions in the next section.

\section{Sequentially Maximizing Solutions}\label{sec:seq max}
In this section, we consider a solution approach for $\mathbf{GP}$ that maximizes the designer's value function backward inductively. In particular, once the value functions have been obtained for times $T,\ldots, t+1$, the designer maximizes $V_t(c_t)$ for each realization of common information $c_t$ subject to the corresponding constraints of $\mathbf{GP}$ at time $t$. This leads to the following definition.
\begin{definition}\label{def:multi agent seq dom}
We say that the tuple $(g^d,\eta_{1:T},W^{1,2}_{1:T}, V_{1:T})$ is  \emph{sequentially maximizing} if it can be obtained backward inductively as follows:\\
   1. \emph{For $t=T$:}  Starting with 
    $$
    W_{T+1}^1(\cdot) = W_{T+1}^2(\cdot) = V_{T+1}(\cdot) = 0,
    $$ 
for each $c_T \in \mathcal C_T$,
    $$
    (\eta_T(\cdot\mid c_T), g^d_T(\cdot\mid \cdot,c_T), W^1_T(c_T), W^2_T(c_T), V_T(c_T))
    $$ 
    is a solution of the following optimization problem $\mathbf{OP_T}(c_T)$:
 \begin{align*}
        &\max_{\eta_T(\cdot\mid c_T),g_T^d(\cdot\mid \cdot,c_T), V_T(c_T), W_T^1(c_T),  W_T^2(c_T)} V_T(c_T)\notag \\
        &\text{\rm s.t.} \quad\eta_T(\cdot\mid c_T), g^d_T(\cdot\mid \cdot, c_T)\text{ satisfy } \eqref{eq:eta multi} \text{ for all } x_T, p_T^{0:2}, m_T^{1,2}, u_T^0, n_T,\\
        &\hspace{22pt}\eta_T(\cdot \mid  c_T)\text{ satisfies }\eqref{eq:multilp}\text{ for all }m_{T}^i, p_{T}^i \text{ for $i\in\{1,2\}$},\\
        &\hspace{22pt}\eta_T(\cdot \mid  c_T), W^i_{T}(c_T)\text{ satisfy \eqref{eq:multi agent value function time any t} for $i\in\{1,2\}$},\\
        &\hspace{22pt}\eta_T(\cdot \mid  c_T), V_{T}(c_T)\text{ satisfy \eqref{eq:designer value function time any t multi}}.
    \end{align*}
  2. \emph{For $t < T$:} For each $c_t \in \mathcal C_t$,   with $W_{t+1}^1(c_t,\cdot)$, $W_{t+1}^2(c_t,\cdot)$, and $V_{t+1}(c_t,\cdot)$ fixed at the values obtained at time $t+1$,  
    $$
    (\eta_t(\cdot\mid c_t), g^d_t(\cdot\mid \cdot,c_t), W^1_t(c_t), W^2_t(c_t), V_t(c_t))
    $$ 
    is a solution of the following optimization problem $\mathbf{OP_t}(c_t)$:
    \begin{align*}
        &\max_{\eta_t(\cdot\mid c_t),g_t^d(\cdot\mid \cdot,c_t), V_t(c_t), W_t^1(c_t),  W_t^2(c_t)} V_t(c_t)\notag \\
        &\text{\rm s.t.} \quad\eta_t(\cdot\mid c_t), g^d_t(\cdot\mid \cdot, c_t)\text{ satisfy } \eqref{eq:eta multi} \text{ for all } x_t, p_t^{0:2}, m_t^{1,2}, u_t^0, n_t,\\
        &\hspace{22pt}\eta_t(\cdot \mid  c_t), W_{t+1}^i(c_t, \cdot) \text{ satisfy }\eqref{eq:multilp} \text{ for all } m_{t}^i, p_{t}^i\text{ for $i\in\{1,2\}$},\\
        &\hspace{22pt}\eta_t(\cdot \mid  c_t), W^i_{t}(c_t), W^i_{t+1}(c_t, \cdot)\text{ satisfy \eqref{eq:multi agent value function time any t} for $i\in\{1,2\}$},\\
        &\hspace{22pt}\eta_t(\cdot \mid  c_t), V_{t}(c_t), V_{t+1}(c_t, \cdot)\text{ satisfy \eqref{eq:designer value function time any t multi}}.
    \end{align*}
\end{definition}

The following lemmas establish that  (i) a sequentially maximizing tuple is a feasible solution of $\mathbf{GP}$, and (ii) the optimization problems in Definition~\ref{def:multi agent seq dom} can be solved as linear programs.
\begin{lemma}\label{lem:seq_feasible}
    A sequentially maximizing tuple is a feasible solution of $\mathbf{GP}.$
\end{lemma}
\begin{proof}
    The proof is immediate since, by construction, a sequentially maximizing tuple satisfies all the constraints of $\mathbf{GP}$.
\end{proof}
\begin{lemma}
\label{lem:lp_multi}
Suppose that $\mathbf{OP}_l(c_l)$ has been solved backward inductively for
$l=T,\ldots,t+1$ and every $c_l\in\mathcal C_l$. Then
$\mathbf{OP}_t(c_t)$ is a linear program for every
$c_t\in\mathcal C_t$.
\end{lemma}
\begin{proof}
At $t=T$,
\[
W_{T+1}^1(\cdot) = W_{T+1}^2(\cdot) = V_{T+1}(\cdot) = 0.
\]
Hence, all terms involving products of $\eta_T$ and future value functions vanish, and the objective and constraints of $\mathbf{OP}_T(c_T)$ are linear in its variables.
Now suppose that $\mathbf{OP}_l(c_l)$ has been solved for
$l=T,\ldots,t+1$ and every $c_l\in\mathcal C_l$. Then
$W_{t+1}^1(c_t,\cdot)$, $W_{t+1}^2(c_t,\cdot)$, and
$V_{t+1}(c_t,\cdot)$ are fixed constants in $\mathbf{OP}_t(c_t)$.
Consequently, \eqref{eq:eta multi}, \eqref{eq:multilp},
\eqref{eq:multi agent value function time any t}, and
\eqref{eq:designer value function time any t multi} are linear in
$\eta_t(\cdot\mid c_t)$, $g_t^d(\cdot\mid\cdot,c_t)$,
$W_t^1(c_t)$, $W_t^2(c_t)$, and $V_t(c_t)$. Therefore, the objective and constraints of $\mathbf{OP}_t(c_t)$ are linear in its variables. 
\end{proof}

Lemma \ref{lem:lp_multi} implies that a sequentially maximizing tuple can be obtained by  a backward inductive sequence of linear programs. Algorithm~\ref{alg:multi agent sequential} implements this backward procedure.
\begin{algorithm}[H]
\caption{}\label{alg:multi agent sequential}
\begin{algorithmic}
    \State{$W^1_{T+1}(\cdot) = W^2_{T+1}(\cdot) = V_{T+1}(\cdot) =0$}
    \For{$t=T,T-1,\ldots,1$}
        \For{ each $c_t \in \mathcal{C}_t$} 
            \State{Solve $\mathbf{{OP}_t}(c_t)$}
            \If{$\mathbf{{OP}_t}(c_t)$ is infeasible}
                        \State{
                            \Return{The selected backward-induction path is infeasible.}
                        }
                    \EndIf
            \EndFor
            \State{$g^d_t = \{g_t^d(\cdot\mid \cdot,c_t)\}_{c_t \in \mathcal{C}_t}, \eta_t = \{\eta_t(\cdot\mid c_t)\}_{c_t \in \mathcal{C}_t}, V_t = \{V_t(c_t)\}_{c_t \in \mathcal{C}_t},W_t^i = \{W_t^i(c_t)\}_{c_t \in \mathcal{C}_t},\:i \in \{1,2\}$}
    \EndFor
    \State{\Return{$(g^d,\eta_{1:T},W^{1,2}_{1:T}, V_{1:T})$}}
\end{algorithmic}
\end{algorithm}
\begin{remark}
When $\mathbf{{OP}_t}(c_t)$ has multiple optimal solutions, different optimizers can select different values of $W_t^1(c_t)$ and $W_t^2(c_t)$. For $\gamma=1$, these value functions enter the constraints in $\mathbf{{OP}_{t-1}}(c_{t-1})$ as continuation values.
Therefore, when $\gamma=1$, if Algorithm~\ref{alg:multi agent sequential} encounters an infeasible  problem at some time $t$,   this shows only that the particular sequence of optimal solutions selected at later times cannot be extended along the current backward induction path. It does not imply that Problem~\ref{prob: multi} or Problem~\ref{prob:feasibility} is necessarily infeasible.
\end{remark}
\begin{remark}
Although the model and the preceding results are stated for two agents, the single-agent case is contained as a special case. Let the private information, message, and action spaces of agent 2 be singleton sets and let \(r_t^2\equiv0\). In this case, agent 2  has only one possible strategy and its CISR constraints are vacuous. Thus,  \(\mathbf{P}(\gamma,r_{1:T}^0)\) reduces directly to its single-agent counterpart. 
\end{remark}
\begin{lemma}\label{lemma:counterexample}
Suppose $\gamma = 1$. 
\begin{enumerate}
\item If $r_t^0 \equiv 0$ for all $t$, the designer-strategy component of a sequentially maximizing tuple returned by Algorithm~\ref{alg:multi agent sequential} solves Problem~\ref{prob:feasibility};
    \item If $r_t^0$ is general, the designer-strategy component of a sequentially maximizing tuple returned by Algorithm~\ref{alg:multi agent sequential} is feasible for Problem~\ref{prob: multi}, but need not be globally optimal.
\end{enumerate}
\end{lemma}
\begin{proof}
By Lemma \ref{lem:seq_feasible}, any sequentially maximizing tuple returned by Algorithm~\ref{alg:multi agent sequential} is feasible for $\mathbf{GP}$. When \(\gamma=1\) and \(r_t^0\equiv0\), the designer-strategy component of a feasible solution of $\mathbf{GP}$ also solves the feasibility problem  $\mathbf{P}(1,r_{1:T}^0\equiv 0)$ which is the same as  Problem~\ref{prob:feasibility}. Similarly, when $\gamma=1$ and $r^0_t$ is general,  the designer-strategy component $g^d$ of a feasible solution of $\mathbf{GP}$ is also feasible  for $\mathbf{P}(1,r_{1:T}^0)$. However, global optimality of $g^d$ returned by Algorithm~\ref{alg:multi agent sequential} need not hold when $r_t^0$ is general. We illustrate it through a counterexample in Appendix~\ref{app:counterexample}. 
\end{proof}

Thus, for non-myopic agents, sequential maximization guarantees feasibility for Problem \ref{prob: multi} and Problem \ref{prob:feasibility},  but does not in general guarantee global optimality for Problem \ref{prob: multi}. The next section shows that a stronger guarantee holds for myopic agents: sequential maximization yields a globally optimal solution to Problem~\ref{prob:myopic agents}.

\section{Myopic Agents: Global Optimality and Belief-State Reduction}\label{sec:optimality_guarantees}
In this section, we focus on Problem~\ref{prob:myopic agents} with myopic agents which corresponds to  using $\gamma =0$ in $\mathbf{P}(\gamma,r_{1:T}^0)$.
\begin{theorem}\label{thm:multi agent special cases}
     Suppose $\gamma = 0$. The designer-strategy component of a sequentially maximizing tuple returned by Algorithm~\ref{alg:multi agent sequential} is globally optimal for Problem~\ref{prob:myopic agents}.
\end{theorem}
\begin{proof}
    See Appendix \ref{app:special cases}. 
\end{proof}
For myopic agents, the sequentially maximizing procedure of Algorithm \ref{alg:multi agent sequential} admits a substantially stronger guarantee than in the non-myopic setting. When $\gamma=0$, the agents' continuation value functions do not enter the current CISR constraints. Consequently, current incentive feasibility is not coupled with the future agents' value functions selected at later stages. This structural simplification allows the backward inductive procedure to recover a globally optimal solution of Problem~\ref{prob:myopic agents}.

\subsection{Computational Considerations}\label{sec:computation}
Examining Algorithm \ref{alg:multi agent sequential}, we observe that at each time $t$, the number of possible common information realizations determines the number of linear programs that must be solved. This number can grow very quickly with time. 
We aim to identify conditions in the case of myopic agents that reduce the number of linear programs to be solved, thereby improving computational efficiency. We will consider the special case where the agent strategies $h^1$ and $h^2$ that the designer wants to incentivize depend on $c_t$ only through the belief $\pi_t$, i.e., for any two realizations of common information $c_t, \hat{c}_t$ such that $\pi_t(\cdot\mid c_t) = \pi_t(\cdot\mid \hat{c}_t),$ we have that $h^i_t(m,p,c_t) = h^i_t(m,p,\hat{c}_t)$ for all $m,p$, and for each $i \in \{1,2\}$. Simple examples are the  \emph{obedient strategies} where $h^i_t(m,p,c_t) = m$.
\begin{theorem}\label{thm:belief comp myopic agent}
Suppose that \(\gamma=0\) and that the prescribed strategies
\(h_t^1,h_t^2\) depend on \(c_t\) only through
\(\pi_t(\cdot\mid c_t)\). Let \(c_t,\hat c_t\) satisfy
\[
    \pi_t(\cdot\mid c_t) = \pi_t(\cdot\mid\hat c_t).
\]
Then \(\mathbf{OP}_t(c_t)\) and \(\mathbf{OP}_t(\hat c_t)\) are identical after the natural relabeling of their decision variables. Consequently, they have the same optimal value, and any optimizer of one problem can be transferred to the other.
\end{theorem}
\begin{proof}
     First consider time $t=T$ and the linear programs $\mathbf{{OP}_T}(c_T)$ and $\mathbf{{OP}_T}(\hat c_T)$. 
      After relabeling the optimization variables, the constraints \eqref{eq:eta multi}, \eqref{eq:multilp}, \eqref{eq:multi agent value function time any t}, and \eqref{eq:designer value function time any t multi} have the same coefficients in $\mathbf{{OP}_T}(c_T)$ and $\mathbf{{OP}_T}(\hat c_T)$. The two linear programs are therefore identical and have the same optimal value. In particular, the optimal designer value function $V_T$ depends on the belief $\pi_T$ and not on the realization of  the common information at time $T$.
We can now proceed inductively. Assume that $V_{t+1}$ depends only on $\pi_{t+1}$. Consider $c_t, \hat{c}_t$ such that $\pi_t(\cdot\mid c_t) = \pi_t(\cdot\mid \hat{c}_t)$. Then, for any realization of the common information increment $z_{t+1}$, the beliefs $\pi_{t+1}(\cdot\mid c_t,z_{t+1})$ and $\pi_{t+1}(\cdot\mid \hat{c}_t,z_{t+1})$ will also be the same. (This is because $\pi_{t+1}$ depends only on $\pi_t$ and the increment $z_{t+1}$, see {\cite[section II.D, Lemma 1 and equation (10)]{nayyar2013common}}). 
Now, since $(c_t,z_{t+1})$ and  $(\hat{c}_t,z_{t+1})$ result in the same $\pi_{t+1}$, the induction hypothesis says that they will have the same designer value functions at $t+1$, i.e.,     $V_{t+1}(c_t, z_{t+1}) =  V_{t+1}(\hat{c}_t, z_{t+1})  $. Moreover, because $\gamma=0$, the  constraints in $\mathbf{{OP}_t}(c_t)$ and $\mathbf{{OP}_t}(\hat c_t)$ do not depend on the agents' future value functions $W_{t+1}^{1,2}$. We can now proceed as we did at time $T$: the linear equalities of \eqref{eq:eta multi}, and the terms multiplying $\eta_t(\cdot)$ in \eqref{eq:multilp}, \eqref{eq:multi agent value function time any t}, \eqref{eq:designer value function time any t multi}  are identical in  $\mathbf{{OP}_t}(c_t)$ and $\mathbf{{OP}_t}(\hat c_t)$ after relabeling. Thus, the two linear programs are identical and will share an optimal solution. 
This completes the induction argument. 
\end{proof}
Under the conditions of Theorem~\ref{thm:belief comp myopic agent}, it is sufficient at each time $t$ to solve one linear program for each distinct belief $\pi_t$, rather than for each realization $c_t$. Thus, whenever several common information realizations induce the same belief, the number of linear programs is reduced. The congestion example below satisfies these conditions and illustrates the computation of a globally optimal designer strategy.

\subsection{Example}\label{sec:example v2}
\emph{The Model:} 
We consider an example with one designer and $K$ myopic agents. The setup described below is a modification of a congestion game example in \cite{tavafoghi2017informational}. The $K$ myopic agents need to travel from an origin to a destination on each day of a $T-$ day horizon. Each agent can choose to take one of two routes. Route $0$ is a safe route associated with a fixed condition  $a$, where $a$ is a positive number  known to all agents and the designer. Route $1$ is a risky route whose condition on day $t$ is described by a random variable $X_t$. (The condition of a route on day $t$ can be interpreted as its traffic capacity on that day) The process  $X_t, t=1,\ldots,T,$ is an uncontrolled Markov chain with state space $\mathcal{X}_t =\{\theta^1,\theta^2\}$ (where $\theta^1,\theta^2$ are non-negative numbers) and initial state distribution given as $\mathbb{P}(X_1 = \theta^1) = p_{\theta^1}, \mathbb{P}(X_1 = \theta^2) = p_{\theta^2}$. The Markov chain evolution can be written as $X_{t+1} = f(X_t, N_t),$ where $N_t \in \mathcal{N}, t \ge 1,$ are iid random variables with distribution $Q$. For $i,j=1,2,$ let $\mathbf{P}(\theta^i,\theta^j)$ denote the   transition probabilities of the Markov chain, i.e., 
$ \mathbf{P}(\theta^i,\theta^j) = \mathbb{P}(X_{t+1} = \theta^j|X_t = \theta^i) = \mathbb{P}(f(\theta^i, N_t)=\theta^j)$. 
The action space for each agent  is given as $\mathcal{U}^i_t = \{0,1\}$ (indicating the two possible routes) for $i=1,\ldots,K, t=1,\ldots,T$. The designer does not take any action (i.e., $\mathcal{U}_t^0 = \emptyset$), its role is just to send messages to each agent based on its information. The designer's message space for each agent is $\mathcal{M}^i_t =\mathcal{U}^i_t = \{0,1\}$. We consider the following information structure. At time $t$, the designer and all agents have full access to the past states and all past actions of the $K$ agents, i.e., $C_t = \{X_{1:t-1}, U_{1:t-1}^{1:K}\}$ (with $C_1 = \emptyset)$). Only the  designer knows the condition of Route $1$ at time $t$, that is, $P_t^0 = \{X_t\}$. Agents don't have any private information i.e., $P_t^i =  \emptyset$ for $i=1,\ldots,K$. At time $t$, the designer generates the messages $M^{1:K}_t$ according to the distribution $g^d_t(X_t, C_t)$.   $g^d_t(m_t^{1:K}|x_t, c_t)$ denotes the probability that messages  $m_t^{1:K}$ are generated  when designer's private and common information at $t$ are $x_t, c_t$, respectively. 

The reward function for agent $i$ at time $t$ is the difference between the condition of the route it chose and the fraction of agents who chose the same route as agent $i$. That is, the reward function can be expressed as follows:
\begin{equation}\label{eq:ex agent reward atomic}
r_t^i(X_t, U_t^{1:K}) =  \begin{cases}
a-\frac{1}{K}\sum_{j=1}^K (1-U_t^j) &\text{if $U_t^i=0$,}\\
&~\\
X_t-\frac{1}{K}\sum_{j=1}^K U_t^j &\text{if $U_t^i=1$.}
\end{cases}
\end{equation}

The designer is interested in social welfare and therefore its reward function is given as: $r_t^0(X_t, U_t^{1:K}) = \sum_{i=1}^K r_t^i(X_t, U_t^{1:K})$.

We are interested in the setting where the designer wants to incentivize \emph{obedience} from all agents, i.e., to incentivize strategies $h^{1:K}$ such that $U_t^i = h_t^i(M_t^i, P_t^i, C_t) = M_t^i$ for  $i=1,\ldots,K$ and $t=1,\ldots,T$, while maximizing its total expected reward  under the strategy profile $(g^d, h^{1:K})$. That is, we have the following problem:
    \begin{align*}
    &\max_{g^d \in \mathcal{G}^d} J^0(g^d, h^{1:K}) := \mathbb{E}^{(g^d,  h^{1:K})}\left[\sum_{t=1}^T r_t^0(X_t, U_t^{1:K})\right] \\
    &\hspace{10pt}\text{s.t. $h^{1:K}$ satisfies $CISR(g^d)$ as per Definition \ref{def:cisr myopic}.}
    \end{align*}

\emph{Globally Optimal Solution Using Algorithm \ref{alg:multi agent sequential}:} 
The information structure described above  satisfies Assumption \ref{assm:1_two_agents}. Further, the beliefs $\pi_t$ are given as:
\begin{align}\label{eq:belief adjusted info structure}
    \pi_1(x_1 | c_1) &= p_{x_1}, \notag\\
    ~\ \pi_t(x_t | x_{1:t-1}, u_{1:t-1}^{1:K}) &= \mathbf{P}(x_{t-1}, x_t)\quad\forall t > 1.
\end{align}
The $\pi_t$ in \eqref{eq:belief adjusted info structure} are strategy-independent, satisfying Assumption \ref{assm:2}. Hence, we can apply  Algorithm \ref{alg:multi agent sequential} (modified for \emph{$K$ agents} instead of two) to find a sequentially maximizing tuple for the designer problem. Furthermore, since the agents are myopic,   the designer-strategy component of a tuple obtained from Algorithm~\ref{alg:multi agent sequential} is also globally optimal (see Theorem~\ref{thm:multi agent special cases}).

According to Algorithm \ref{alg:multi agent sequential}, at each time $t$, we need to solve a linear program for each possible realization of $x_{1:t-1}, u_{1:t-1}^{1:K}$. This implies we have $\mathcal{O}(2^{Kt})$ linear programs to solve. However, we can make use of Theorem \ref{thm:belief comp myopic agent} to drastically reduce the number of linear programs.  Observe that the obedient strategies $h^i$ do not depend on $c_t$ at all and that the  $\pi_t$ in \eqref{eq:belief adjusted info structure} depends only on $x_{t-1}$ and not on $x_{1:t-2}, u_{1:t-1}^{1:K}$. Therefore,  by Theorem \ref{thm:belief comp myopic agent}, any two realizations of $c_t$ with the same $x_{t-1}$ will effectively result in the same linear program. Thus, the number of linear programs to be solved at time $t$ is just the number of possible realizations of $x_{t-1}$, namely $2$.

\emph{Numerical results:} We implemented Algorithm \ref{alg:multi agent sequential} in Matlab for the following parameters: $K = 10, T=2, a= 1.5, \theta^1 = 1.2, \theta^2 = 2.8, p_{\theta^1} = p_{\theta^2} = 0.5, \mathbf{P}(\theta^1, \theta^1) = \mathbf{P}(\theta^2, \theta^2) = 0.9$.  We observed that for each $x_t, c_t,$ $g^d_t(m_t^{1:K}\mid x_t,c_t)$ depends on $\Sigma_t = \sum_{i=1}^K m_t^i$.
We obtain the following  messaging strategy for the designer: \\
1. At $t=1$,    $g_1^d(m_1^{1:10} \mid  x_1 = 1.2) = 0.0048$ if $\Sigma_1 = 4$ and $0$ otherwise; and
     $g_1^d(m_1^{1:10} \mid  x_1 = 2.8) = 0.0222$ if $\Sigma_1 = 8$ and $0$ otherwise.\\
2. At $t=2$,
\begin{itemize}
\setlength{\itemsep}{5pt} 
    \item $g_2^d(m_2^{1:10} \mid  x_2 = 1.2, x_1 = 1.2, u_1^{1:10}) = 0.0048$ if $\Sigma_2 = 4$ and $0$ otherwise.
    \item $g_2^d(m_2^{1:10} \mid  x_2 = 2.8, x_1 = 1.2, u_1^{1:10}) = 0.0222$ if $\Sigma_2 = 8$ and $0$ otherwise.
    \item $g_2^d(m_2^{1:10} \mid  x_2 = 1.2, x_1 = 2.8, u_1^{1:10}) = 0.0039$ if $\Sigma_2 = 5$;  $0.0001$ if $\Sigma_2 = 6$ and $0$ otherwise.
    \item $g_2^d(m_2^{1:10} \mid  x_2 = 2.8, x_1 = 2.8, u_1^{1:10}) = 0.0557$ if $\Sigma_2 = 9$; $0.4426$ if $\Sigma_2 = 10$ and $0$ otherwise.
   \end{itemize}

\section{Discussion and Conclusion}\label{sec:conclusion}
\subsection{Relationship to equilibrium concepts}\label{sec:discussion on cisr}
The prescribed strategies \(h^1\) and \(h^2\) represent the behavior that the designer seeks to induce. In the general asymmetric-information setting we considered, an agent's desired action may depend on its private information \(P_t^i\), which is not available to the designer. Thus, a prescribed strategy \(h_t^i(M_t^i,P_t^i,C_t)\) allows the designer's message to be combined with the agent's private and common information in determining the desired action; obedience, \(U_t^i=h_t^i(M_t^i,P_t^i,C_t)=M_t^i\), is a special case. The CISR constraints determine whether the prescribed strategies are sequentially  rational (myopically or non-myopically) under a given designer strategy. In the discussion below, we will focus on the obedient strategy special case to connect CISR with other existing equilibrium concepts.

\emph{Static case ($T=1$):} Consider the static setting where the designer observes $X_1$, there is no common information, and the agents have no private information. The designer generates a message pair $(M^1_1, M^2_1) \sim g^d_1(X_1)$ (taking no action $U^0_1$), and $h^i$ is the obedient strategy, i.e., $h^i_1(m^i_1) = m^i_1$. Since $T=1$, there are no future rewards to consider. In this special case inequalities \eqref{eq:multilp} which characterize CISR become 
\begin{align}\label{eq:cisr statc 2}
    &\sum_{x \in \mathcal X_1, m^j \in \mathcal M_1^j} \pi_1(x)g_1^d(m^j, m_1^i \mid x)r_1^i(x, m^j, m_1^i)\geq \sum_{x \in \mathcal X_1, m^j \in \mathcal M_1^j} \pi_1(x)g_1^d(m^j, m_1^i \mid x)r_1^i(x, m^j, u),~\forall {u}\in \mathcal{U}_1^i.
\end{align}
The above inequalities can be related to Bayesian Nash equilibrium (BNE)\cite[Chapter 2]{osborne1994course} as follows:
Once the designer has selected $g^d$, it induces a Bayesian game among the two agents. We refer to this game as $BG(g^d)$. In this game, agent $i$ observes $M^i_1$ and takes action $U^i_1$.  The inequalities in \eqref{eq:cisr statc 2} are precisely the conditions needed for obedient strategies to be a Bayesian Nash equilibrium (BNE) of $BG(g^d)$.

Alternatively, the requirement that obedient strategies satisfy $CISR(g^d)$ can be interpreted as $g^d_1$ being a Bayes correlated equilibrium (BCE) \cite{bergemann2016information}.
To see this, we simply note that 
\eqref{eq:cisr statc 2} are precisely the BCE obedience constraints for $g^d_1$ in the static setting described above (see BCE definition in \cite{bergemann2016information}). 

\emph{Dynamic case ($T > 1$):} We first consider the case of myopic agents in the setting where at each time $t$,  the designer observes $X_t$, the agents have no private information, the designer has no action $U_t^0$, and $h^i$ are the obedient strategies. Consider a designer strategy $g^d$ and 
fix a  realization $c_t$ of the common information. Conditional on $C_t=c_t$, the myopic CISR conditions in \eqref{eq:cisr myopic} reduce to 
\begin{align}\label{eq:cisr statc 3}
    &\sum_{x \in \mathcal X_t, m^j \in \mathcal M_t^j} \pi_t(x \mid c_t)g_t^d(m^j, m_t^i \mid x, c_t)r_t^i(x, m^j, m_t^i)\geq \sum_{x \in \mathcal X_t, m^j \in \mathcal M_t^j} \pi_t(x \mid c_t)g_t^d(m^j, m_t^i \mid x, c_t)r_t^i(x, m^j, u),~\forall {u}\in \mathcal{U}_t^i.
\end{align}
These conditions are similar to \eqref{eq:cisr statc 2} in the static case and can be interpreted as follows: for a fixed $g^d$ and $c_t$, the two myopic agents are engaged in a one-shot Bayesian game (which we will denote by $BG(g^d;c_t)$). The CISR conditions ensure that obedience is a BNE of this one-shot game. In other words, CISR ensures  that for each realization of the common information, obedience is a myopic BNE for the agents.

For non-myopic agents, our CISR conditions are conceptually closely related to the common information based Markov perfect equilibrium (CIB-MPE) in dynamic stochastic games \cite{nayyar2013common}. Both  CISR and CIB-MPE rely on the strategy-independence of the common information based conditional belief $\pi_t(\cdot|c_t)$. A key equilibrium idea behind CIB-MPE is that at each time $t$ and for each realization of $c_t$, each agent's strategy maximizes its expected reward-to-go conditioned on $c_t$. This is similar to the idea behind non-myopic CISR - at each time $t$ and for each $c_t$, the prescribed strategy $h^i$ for agent $i$ maximizes its expected reward-to-go conditioned on~$c_t$.

\subsection{Conclusion}
We considered a dynamic information design problem in which a designer jointly selects private messages and its own actions in order to incentivize multiple agents to follow prescribed strategies. Under strategy-independent common information based beliefs, we developed an exact reformulation of the designer's problem.  Motivated by the structure of this reformulation, we identified a sub-class of feasible solutions (referred to as sequentially maximizing tuples) that can be obtained by a backward-inductive procedure. This procedure relies on solving a family of linear programs - one for each time and each realization of common information.   For non-myopic agents, this procedure guarantees feasibility  but need not provide a globally optimal solution for a general designer reward. For myopic agents, we showed that the backward procedure provides a globally optimal solution and, when the prescribed strategies depend on common information only through the belief $\pi_t$, the computation can be reduced to one linear program for each distinct belief. We illustrated these results through a dynamic  congestion game example. The backward-inductive and linear-programming structure of our approach relies critically on the assumed information structure properties; more general information structures may lead to strategy-dependent beliefs and substantially more difficult optimization problems.


\appendices
\section{Proof of Lemma \ref{lem:suff nece multi}}\label{app:multi equiv}
\emph{Sufficiency:} Assume that the condition in \eqref{eq:suff1 multi} of Lemma \ref{lem:suff nece multi} holds. We show that \eqref{eq:cisr} holds for each $t$ and each realization $c_t$ of the common information. Fix $i \in \{1,2\}$. For convenience, define
$$
\mathcal{S}^i_t(c_t):= \{(p_t^i, m_t^i): \eta_t(p_t^i, m_t^i \mid  c_t) > 0\}.
$$

We first consider time $t=T$ and an arbitrary  $c_T \in \mathcal{C}_T$. In this case, the right-hand side of \eqref{eq:cisr} can be written as (we omit the superscript $(g^d, g^i, h^j)_T$ in some of the expectations below)
\begin{align}
    \mathbb{E}&^{(g^d, g^i, h^j)_T}\left[ r_T^i(X_T, U_T^{0,j,i}) \mid  c_T\right] \notag \\
    &=\mathbb{E}\Big[ \mathbb{E}[r_T^i(X_T, U_T^{0,j},g^i_T(M_T^i,P^i_T,c_T))\mid c_T,P^i_T,M_T^i] \Bigm| c_T\Big] \notag \\
    &=\sum_{(p^i_T,m_T^i) \in \mathcal S_T^i(c_T)} \Big[ \eta_T(p^i_T,m_T^i\mid c_T) \times\mathbb{E}^{\mu_T^i(\cdot\mid m^i_T,p^i_T,c_T)}[r_T^i(X_T,U_T^{0,j}, g^i_T(m_T^i,p^i_T,c_T))] \Big] \notag \\
    &\le\sum_{(p^i_T,m_T^i)\in\mathcal S^i_T(c_T)} \Big[ \eta_T(p^i_T,m_T^i\mid c_T)\times\mathbb{E}^{\mu_T^i(\cdot\mid m^i_T,p^i_T,c_T)}[r_T^i(X_T,U_T^{0,j}, h^i_T(m_T^i,p^i_T,c_T))] \Big],\label{eq:app1eq1}
\end{align}
where $U_T^j = h^j_T(M_T^j,P^j_T,c_T)$. The second equality follows from the definitions of \(\eta_T\) and \(\mu_T^i\), with the sum restricted to \(\mathcal{S}_T^i(c_T)\). 
The inequality above follows from
\eqref{eq:suff1 multi}, together with \(W_{T+1}^i\equiv0\). Repeating the above steps with $h^i$ instead of $g^i$ will result in 
\begin{align}
    &\mathbb{E}^{(g^d, h^i, h^j)_T}\left[ r_T^i(X_T, U_T^{0,j,i}) \mid  c_T\right] =\sum_{(p^i_T,m_T^i) \in \mathcal S_T^i(c_T)} \Big[ \eta_T(p^i_T,m_T^i\mid c_T) \times\mathbb{E}^{\mu_T^i(\cdot\mid m^i_T,p^i_T,c_T)}[r_T^i(X_T,U_T^{0,j}, h^i_T(m_T^i,p^i_T,c_T))] \Big].  \label{eq:app1eq2} 
\end{align}
\eqref{eq:app1eq1} and \eqref{eq:app1eq2} establish \eqref{eq:cisr} for time $T$. 

We can now proceed inductively. Assume that \eqref{eq:cisr} holds for $t+1$ and each $c_{t+1}$. By \eqref{eq:W-prescribed-payoff}, \(W_{t+1}^i(c_{t+1})\) is equal to the left hand side of \eqref{eq:cisr} for time $t+1$ and each $c_{t+1}$.
Consider time $t$ and any $c_t \in \mathcal{C}_t$.
In this case, the right-hand side of \eqref{eq:cisr} can be written as (we omit the superscript $(g^d, g^i, h^j)_{t}$ in some of the expectations below for convenience)
\begin{align}
    &\mathbb{E}^{(g^d, g^i, h^j)_{t:T}}\left[ \sum_{k=t}^T \gamma^{k-t} r_k^i(X_k, U_k^{0,j,i}) \Bigm| c_t\right] \notag \\
    &= \mathbb{E}^{(g^d, g^i, h^j)_{t}}\left[r_t^i(X_t, U_t^{0,j,i}) \mid  c_t\right]+ \mathbb{E}\left[\mathbb{E}^{(g^d, g^i, h^j)_{t+1:T}}\left[ \sum_{k=t+1}^T \gamma^{k-t} r_k^i(X_k, U_k^{0,j,i}) \Bigm| C_{t+1}\right]  \Bigm| c_t\right] \notag \\
 &\le \mathbb{E}^{(g^d, g^i, h^j)_{t}}\left[r_t^i(X_t, U_t^{0,j,i}) \mid  c_t\right]+ \mathbb{E}\left[\mathbb{E}^{(g^d, h^i, h^j)_{t+1:T}}\left[ \sum_{k=t+1}^T \gamma^{k-t} r_k^i(X_k, U_k^{0,j,i}) \Bigm| C_{t+1}\right]  \Bigm| c_t\right] \notag \\
    &= \mathbb{E}^{(g^d, g^i, h^j)_{t}}\left[r_t^i(X_t, U_t^{0,j,i}) + \gamma W^i_{t+1}(c_t, Z_{t+1}) \Bigm| c_t\right], \label{eq:app1eq3}
\end{align}
where we used the induction hypothesis in the inequality above. For notational convenience, let 
\begin{align}
&F^i(c_t, x_t,p^{0,j,i}_t, u^{0,j,i}_t, m_t^{j, i}, n_t) = r_t^i(x_t, u_t^{0, j, i}) + \gamma W^i_{t+1}(c_t, z_{t+1}),
\end{align}
where $z_{t+1} = \zeta_{t+1}(x_t, p_t^{0,j,i}, u_t^{0,j,i}, m_t^{j, i}, n_{t})$. 
Then, the right-hand side of \eqref{eq:app1eq3} can be written as 
\begin{align}
   & \mathbb{E}^{(g^d, g^i, h^j)_{t}}\left[ F^i(c_t, X_t,P^{0,j,i}_t, U^{0,j,i}_t, M_t^{j,i}, N_t)  \mid  c_t\right] \notag \\
    &=\mathbb{E}\Big[ \mathbb{E}[F^i(c_t, X_t,P^{0,j,i}_t,U^0_t, h^j_t(M^j_t,P^j_t,c_t), g^i_t(M^i_t,P^i_t,c_t), M_t^{j,i}, N_t)\mid c_t,P^i_t,M_t^i] \Bigm| c_t\Big] \notag \\
    &=\sum_{(p^i_t, m_t^i) \in \mathcal S_t^i(c_t)} \Big[ \eta_t(p^i_t,m_t^i\mid c_t)\times \mathbb{E}^{\mu_t^i(\cdot\mid m^i_t,p^i_t,c_t)}[F^i(c_t, X_t, P^{0,j}_t,p^i_t,U^0_t,h^j_t(M_t^j,P^j_t,c_t),g^i_t(m_t^i,p^i_t,c_t), M_t^j, m_t^{i}, N_t) ] \Big] \label{eq:app1eq4} \\
    &\le\sum_{(p^i_t, m_t^i) \in \mathcal S_t^i(c_t)} \Big[ \eta_t(p^i_t,m_t^i\mid c_t)\times \mathbb{E}^{\mu_t^i(\cdot\mid m^i_t,p^i_t,c_t)}[F^i(c_t, X_t, P^{0,j}_t,p^i_t,U^0_t,h^j_t(M_t^j,P^j_t,c_t),h^i_t(m_t^i,p^i_t,c_t), M_t^j, m_t^{i}, N_t) ] \Big],\label{eq:app1eq5}
\end{align}
where we used \eqref{eq:suff1 multi} and the definition of $F^i$ in the last inequality. Combining \eqref{eq:app1eq3} and \eqref{eq:app1eq5}, we obtain
\begin{align}
 &\mathbb{E}^{(g^d, g^i, h^j)_{t:T}}\left[ \sum_{k=t}^T \gamma^{k-t} r_k^i(X_k, U_k^{0,j,i}) \Bigm| c_t\right] \notag \\
    &\le\sum_{(p^i_t, m_t^i) \in \mathcal S_t^i(c_t)} \Big[ \eta_t(p^i_t,m_t^i\mid c_t)\times \mathbb{E}^{\mu_t^i(\cdot\mid m^i_t,p^i_t,c_t)}[F^i(c_t, X_t, P^{0,j}_t,p^i_t,U^0_t,h^j_t(M_t^j,P^j_t,c_t),h^i_t(m_t^i,p^i_t,c_t), M_t^j, m_t^{i}, N_t) ] \Big]. \label{eq:app1eq6}
\end{align}
Repeating the above steps with $h^i$ instead of $g^i$ will result in
\begin{align}
     &\mathbb{E}^{(g^d, h^i, h^j)_{t:T}}\left[ \sum_{k=t}^T \gamma^{k-t} r_k^i(X_k, U_k^{0,j,i}) \Bigm| c_t\right] \notag \\
     &=\sum_{(p^i_t, m_t^i) \in \mathcal S_t^i(c_t)} \Big[ \eta_t(p^i_t,m_t^i\mid c_t)\times \mathbb{E}^{\mu_t^i(\cdot\mid m^i_t,p^i_t,c_t)}[F^i(c_t, X_t, P^{0,j}_t,p^i_t,U^0_t,h^j_t(M_t^j,P^j_t,c_t),h^i_t(m_t^i,p^i_t,c_t), M_t^j, m_t^{i}, N_t) ] \Big]. \label{eq:app1eq7}
     \end{align}
\eqref{eq:app1eq6} and \eqref{eq:app1eq7} establish \eqref{eq:cisr} for time $t$. 
This completes the induction argument and 
therefore establishes \eqref{eq:cisr} for every \(t\) and every
\(c_t\). 

\emph{Necessity:} We provide a proof-by-contradiction argument. Let $i \in \{1,2\}$. Suppose \eqref{eq:cisr} is satisfied but \eqref{eq:suff1 multi} is not true for some time $t$, $1 \le t \le T$, and some realizations $c_t, m_t^i,p_t^i $ with $\eta_t(p_t^i, m_t^i\mid  c_t) > 0$. Let  ${\tau}$ be the largest time index less than or equal to $T$ such that there exists a $\tilde{c}_{\tau}, \tilde{m}_{\tau}^i,\tilde{p}_{\tau}^i $ with $\eta_{\tau}(\tilde{p}_{\tau}^i, \tilde{m}_{\tau}^i\mid  \tilde{c}_{\tau}) > 0$  where \eqref{eq:suff1 multi} is not true. Because the horizon is finite and the set of times at which \eqref{eq:suff1 multi} fails is nonempty, the largest such time \(\tau\) is well defined. By the maximality of \(\tau\), the condition in \eqref{eq:suff1 multi} holds at every time strictly greater than \(\tau\). We will show that each possible value of $\tau$ results in a contradiction.

 Suppose that  $\tau = T$. In this case, we will construct an agent strategy $g^i$ such that \eqref{eq:cisr} is violated.  Consider the agent action strategy $g^i$ that is identical to $h^i$ everywhere except for  realization $\tilde{c}_{T}, \tilde{m}_{T}^i,\tilde{p}_{T}^i $ of the agent's information at  time $T$.
Define $g_T^i(\tilde{m}_{T}^i,\tilde{p}_{T}^i,\tilde{c}_{T})$ as follows
\begin{align}\label{eq:app2eq1}
&g_T^i(\tilde{m}_{T}^i,\tilde{p}_{T}^i,\tilde{c}_{T})\in   \argmax_{u \in \mathcal{U}_T^i} \mathbb{E}^{\mu_T^i(\cdot\mid  \tilde{m}_{T}^i,\tilde{p}_{T}^i,\tilde{c}_T)}[r_T^i(X_T,U_T^0, h_T^j(M_T^j, P_T^j,\tilde{c}_T),u)].
\end{align}

For the $g^i$ defined above and the common information realization $\tilde{c}_T$, the right-hand side of \eqref{eq:cisr} can be written as 
\begin{align}
    \mathbb{E}&^{(g^d, g^i, h^j)_T}\left[ r_T^i(X_T, U_T^{0,j,i}) \mid  \tilde{c}_T\right] \notag \\
    &=\mathbb{E}\Big[ \mathbb{E}[r_T^i(X_T, U_T^{0,j}, g^i_T(M_T^i,P^i_T,\tilde{c}_T))\mid \tilde{c}_T,P^i_T,M_T^i] \Bigm| \tilde{c}_T\Big] \notag \\
    &=\sum_{(p^i_T,m_T^i) \in \mathcal S_T^i(\tilde c_T)} \Big[ \eta_T(p^i_T,m_T^i\mid \tilde{c}_T)\times \mathbb{E}^{\mu_T^i(\cdot\mid m^i_T,p^i_T,\tilde{c}_T)}[r_T^i(X_T, U_T^{0,j}, g^i_T(m_T^i,p^i_T,\tilde{c}_T))] \Big] \notag \\
    &> \sum_{(p^i_T,m_T^i) \in \mathcal S_T^i(\tilde c_T)} \Big[ \eta_T(p^i_T,m_T^i\mid \tilde{c}_T)\times \mathbb{E}^{\mu_T^i(\cdot\mid m^i_T,p^i_T,\tilde{c}_T)}[r_T^i(X_T, U_T^{0,j}, h^i_T(m_T^i,p^i_T,\tilde{c}_T))] \Big]\label{eq:app2.a} \\
    &=\mathbb{E}^{(g^d, h^i, h^j)_T}\left[ r_T^i(X_T, U_T^{0,j,i}) \mid  \tilde{c}_T\right].
    \label{eq:app2.b}
    \end{align}
The strict inequality in \eqref{eq:app2.a} holds because $g^i$ and $h^i$ are identical everywhere except at the realization $\tilde{c}_{T}, \tilde{m}_{T}^i,\tilde{p}_{T}^i$ of agent's information, and for this critical realization we have 
\begin{align}\label{eq:app2eq3}
    &\eta_T(\tilde{p}^i_T, \tilde{m}_T^i\mid \tilde{c}_T) \times 
        \mathbb{E}^{\mu_T^i(\cdot\mid  \tilde{m}_{T}^i,\tilde{p}_{T}^i,\tilde{c}_T)}[r_T^i(X_T, U_T^{0,j}, g^i_T(\tilde{m}_T^i,\tilde{p}^i_T,\tilde{c}_T))]\notag \\
        &\hspace{40pt}> \eta_T(\tilde{p}^i_T, \tilde{m}_T^i\mid \tilde{c}_T) \times 
        \mathbb{E}^{\mu_T^i(\cdot\mid  \tilde{m}_{T}^i,\tilde{p}_{T}^i,\tilde{c}_T)}[r_T^i(X_T, U_T^{0,j}, h^i_T(\tilde{m}_T^i,\tilde{p}^i_T,\tilde{c}_T))],
\end{align}
since $\eta_T(\tilde{p}^i_T, \tilde{m}_T^i\mid \tilde{c}_T) >0$ and $h_T^i$ is not a maximizer of the right-hand side of \eqref{eq:suff1 multi} for time ${T}$ and the given realization $\tilde{c}_{T}, \tilde{m}_{T}^i,\tilde{p}_{T}^i$. 

Thus, \eqref{eq:app2.b} shows that the strategy $g^i$ constructed above violates \eqref{eq:cisr}. Thus, we must have $\tau \ne T$.

We can now consider other possible values of $\tau$. Suppose that $\tau =l$, where $1 \le l <T$ and we have $\tilde{c}_{l}, \tilde{m}_{l}^i,\tilde{p}_{l}^i$ with $\eta_{l}(\tilde{p}_{l}^i, \tilde{m}_{l}^i\mid  \tilde{c}_{l}) > 0$  where \eqref{eq:suff1 multi} is not true. Consider an agent strategy $g^i$ that is identical to $h^i$ everywhere except for the realization $\tilde{c}_{l}, \tilde{m}_{l}^i,\tilde{p}_{l}^i$ of agent $i$'s information at  time $l$. Define $g_l^i(\tilde{m}_{l}^i,\tilde{p}_{l}^i,\tilde{c}_{l})$ 
as follows
\begin{align}
    &g_l^i(\tilde{m}_{l}^i,\tilde{p}_{l}^i,\tilde{c}_{l}) \in \argmax_{u \in \mathcal{U}_l^i} \mathbb{E}^{\mu_l^i(\cdot\mid \tilde{m}_{l}^i,\tilde{p}_{l}^i,\tilde{c}_{l})}[r_l^i(X_l, U_l^{0,j},  u) + \gamma W^i_{l+1}(\tilde c_l, Z_{l+1})], 
\end{align}
where $Z_{l+1} = \zeta_{l+1}(X_l, P_l^{0,j}, \tilde{p}_{l}^i, U_l^{0,j} , u, M_l^j, \tilde m_l^{i}, N_{l})$ with $U_l^j = h_l^j(M_l^j,P_l^j,\tilde c_l)$.

 For the $g^i$ defined above and the common information realization $\tilde{c}_l$, we can follow the steps used in \eqref{eq:app1eq3} to write
\begin{align}
    &\mathbb{E}^{(g^d, g^i, h^j)_{l:T}}\left[ \sum_{k=l}^T \gamma^{k-l} r_k^i(X_k, U_k^{0,j,i}) \Bigm| \tilde{c}_l\right] \notag \\
    &= \mathbb{E}\left[r_l^i(X_l, U_l^{0,j,i}) \mid  \tilde{c}_l \right] + \mathbb{E}\left[\mathbb{E}^{(g^d, g^i, h^j)_{l+1:T}}\Bigg[ \sum_{k=l+1}^T \gamma^{k-l} r_k^i(X_k, U_k^{0,j,i}) \Bigm| C_{l+1}\Bigg]  \Bigm| \tilde{c}_l\right] \notag \\
    &= \mathbb{E}\left[r_l^i(X_l, U_l^{0,j,i}) \mid  \tilde{c}_l \right] + \mathbb{E}\left[ \mathbb{E}^{(g^d, h^i, h^j)_{l+1:T}}\left[ \sum_{k=l+1}^T \gamma^{k-l} r_k^i(X_k, U_k^{0,j,i}) \Bigm| C_{l+1}\right]  \Bigm| \tilde{c}_l\right] \label{eq:app2.d} \\
    &= \mathbb{E}^{(g^d, g^i, h^j)_{l:T}}\left[r_l^i(X_l, U_l^{0,j,i}) + \gamma W^i_{l+1}(\tilde c_l, Z_{l+1}) \Bigm| \tilde{c}_l\right], \label{eq:app2.c}
\end{align}
where we used the fact that $g^i$ and $h^i$ are identical for time $k > l$ in \eqref{eq:app2.d} and the fact that $W^i_{l+1}(c_{l+1})$ is the left-hand side of \eqref{eq:cisr} for each $c_{l+1}$ by \eqref{eq:W-prescribed-payoff}.
Now, following the steps  used to get  \eqref{eq:app1eq5} from \eqref{eq:app1eq3}, \eqref{eq:app2.c} can be written as
\begin{align}\label{eq:app2eq5}
    &\mathbb{E}^{(g^d,g^i,h^j)_{l:T}}
    \left[
        r_l^i(X_l,U_l^{0,j,i})
        +\gamma W_{l+1}^i(\tilde c_l,Z_{l+1})
        \Bigm|\tilde c_l
    \right]
    \notag\\
    &=\sum_{(p^i_l, m_l^i) \in \mathcal S_l^i(\tilde c_l)} \Big[ \eta_l(p^i_l,m_l^i\mid \tilde c_l)\times \mathbb{E}^{\mu_l^i(\cdot\mid m^i_l,p^i_l,\tilde c_l)}[F^i(\tilde c_l, X_l, P^{0,j}_l, p^i_l,U^0_l,h^j_l(M_l^j,P^j_l,\tilde c_l),g^i_l(m_l^i,p^i_l,\tilde c_l), M_l^j, m_l^i, N_l) ] \Big].
\end{align}
Using \eqref{eq:app2.c}, \eqref{eq:app2eq5}  and  arguments similar to those  in \eqref{eq:app2.a} and \eqref{eq:app2.b}, we obtain
\begin{align}
  &\mathbb{E}^{(g^d, g^i, h^j)_{l:T}}\left[ \sum_{k=l}^T \gamma^{k-l} r_k^i(X_k, U_k^{0,j,i}) \Bigm| \tilde{c}_l\right]  \notag \\
  &=\sum_{(p^i_l, m_l^i) \in \mathcal S_l^i(\tilde c_l)} \Big[ \eta_l(p^i_l,m_l^i\mid \tilde c_l)\times \mathbb{E}^{\mu_l^i(\cdot\mid m^i_l,p^i_l,\tilde c_l)}[F^i(\tilde c_l, X_l, P^{0,j}_l,p^i_l,U^0_l,h^j_l(M_l^j,P^j_l,\tilde c_l),g^i_l(m_l^i,p^i_l,\tilde c_l), M_l^j, m_l^i, N_l) ] \Big]\notag\\
  &>\sum_{(p^i_l, m_l^i) \in \mathcal S_l^i(\tilde c_l)} \Big[ \eta_l(p^i_l,m_l^i\mid \tilde c_l)\times \mathbb{E}^{\mu_l^i(\cdot\mid m^i_l,p^i_l,\tilde c_l)}[F^i(\tilde c_l, X_l, P^{0,j}_l,p^i_l,U^0_l,h^j_l(M_l^j,P^j_l,\tilde c_l),h^i_l(m_l^i,p^i_l,\tilde c_l), M_l^j, m_l^i, N_l) ] \Big] \notag \\
  &=\mathbb{E}^{(g^d, h^i, h^j)_{l:T}}\left[ \sum_{k=l}^T \gamma^{k-l} r_k^i(X_k, U_k^{0,j,i}) \Bigm| \tilde{c}_l\right]. \label{eq:app2.e}
\end{align}
\eqref{eq:app2.e} shows that the strategy $g^i$ constructed above violates \eqref{eq:cisr}. Thus, we must have that $\tau \ne l$. 

The above argument shows that $\tau$ cannot take any value in $\{T,T-1,\ldots,1\}$. Since $i$ is chosen arbitrarily, \eqref{eq:suff1 multi} must hold for each agent $i$, for each $t$ and each $c_t, m_t^i,p_t^i $ with $\eta_t(p_t^i, m_t^i\mid  c_t) > 0$.

\section{Proof of Lemma~\ref{lemma:counterexample}}
\label{app:counterexample}
Since the single-agent model is a special case, a single-agent counterexample is sufficient to show that the solution returned by Algorithm~\ref{alg:multi agent sequential} need not be globally optimal for $\mathbf{P}(1,r_{1:T}^0)$.

Consider the following single-agent example with a
two-stage horizon $T=2$. The initial state is deterministic,
$\mathcal X_1=\{x_1\}$, while $\mathcal X_2=\{L,H\}$. The action and message spaces are $\mathcal U_1^1=\mathcal U_2^1=\mathcal M_1^1=\mathcal M_2^1=\{0,1\}$. The designer has no control action, i.e., $\mathcal U_t^0=\{0\}$. The prescribed strategy is obedience, i.e., $h_t^1(M_t^1)=M_t^1$, for $t\in\{1,2\}$.

At time $2$, the designer observes $X_2$, whereas the agent does not. The common information is $C_1=\emptyset$ and $C_2=U_1^1$.
Let $N_1\in\{L,H\}$ with $\mathbb P(N_1=L)=\mathbb P(N_1=H)=1/2$, and let
\[
X_2= f(U_1^1, N_1) = 
\begin{cases}
L, & U_1^1=0,\\
N_1, & U_1^1=1.
\end{cases}
\]
Thus, $X_2=L$ when $U_1^1=0$, while $X_2$ is uniformly distributed over $\{L,H\}$ when $U_1^1=1$. The rewards are
\[
\begin{array}{c|cc}
 & U_1^1=0 & U_1^1=1 \\ \hline
r_1^0 & 0 & 100 \\
r_1^1 & 4 & 0
\end{array}
\qquad
\begin{array}{c|cc}
 & U_2^1=0 & U_2^1=1 \\ \hline
r_2^0 & 0 & 20 \\
r_2^1(L,\cdot) & 0 & -5 \\
r_2^1(H,\cdot) & 0 & 10
\end{array}
\]
and we take $\gamma=1$.

\emph{Strategy returned by Algorithm~\ref{alg:multi agent sequential}:}
At time $t=2$ with $c_2=0$, Algorithm~\ref{alg:multi agent sequential} returns
\[
g_2^{d,\mathrm{alg}}(0\mid L,0) = g_2^{d,\mathrm{alg}}(0\mid H,0) =1,
\]
with
\[
V_2^{\mathrm{alg}}(0) = W_2^{1,\mathrm{alg}}(0) =0.
\]
At time $t=2$ with $c_2=1$, the algorithm returns
\[
g_2^{d,\mathrm{alg}}(1\mid L,1) = g_2^{d,\mathrm{alg}}(1\mid H,1) =1,
\]
with
\[
V_2^{\mathrm{alg}}(1)=20, \qquad W_2^{1,\mathrm{alg}}(1)=\frac52.
\]
At time $1$, the algorithm returns
\[
g_1^{d,\mathrm{alg}}(0\mid x_1)=1,
\]
and
\[
V_1^{\mathrm{alg}}=0, \qquad W_1^{1,\mathrm{alg}}=4.
\]
The corresponding objective value in
$\mathbf{GP}$ is $V_1^{\mathrm{alg}}=0$.

\emph{An alternative feasible solution:}
Define $(g^{d,\mathrm{alt}},\eta^{\mathrm{alt}},\allowbreak
W^{1,\mathrm{alt}},\allowbreak V^{\mathrm{alt}})$ as follows. At time $t=2$ with $c_2=0$,
\[
g_2^{d,\mathrm{alt}}(0\mid L,0) = g_2^{d,\mathrm{alt}}(0\mid H,0) = 1,
\]
\[
\eta_2^{\mathrm{alt}}(L,0\mid0)=1,
\]
\[
\eta_2^{\mathrm{alt}}(L,1\mid0) = \eta_2^{\mathrm{alt}}(H,0\mid0) = \eta_2^{\mathrm{alt}}(H,1\mid0) =0,
\]
and
\[
V_2^{\mathrm{alt}}(0) = W_2^{1,\mathrm{alt}}(0) =0.
\]
At time $t=2$ with $c_2=1$,
\[
g_2^{d,\mathrm{alt}}(0\mid L,1)= g_2^{d,\mathrm{alt}}(1\mid H,1)=1,
\]
\[
\eta_2^{\mathrm{alt}}(L,0\mid1) = \eta_2^{\mathrm{alt}}(H,1\mid1) = \frac12, \quad \eta_2^{\mathrm{alt}}(L,1\mid1) =
\eta_2^{\mathrm{alt}}(H,0\mid1) =0,
\]
and
\[
V_2^{\mathrm{alt}}(1)=10, \qquad W_2^{1,\mathrm{alt}}(1)=5.
\]
At time $1$,
\[
g_1^{d,\mathrm{alt}}(1\mid x_1)=1,
\]
\[
\eta_1^{\mathrm{alt}}(x_1,1,L) = \eta_1^{\mathrm{alt}}(x_1,1,H) = \frac12,
\]
while
\[
\eta_1^{\mathrm{alt}}(x_1,0,L) = \eta_1^{\mathrm{alt}}(x_1,0,H) =0,
\]
and
\[
V_1^{\mathrm{alt}}=110, \qquad W_1^{1,\mathrm{alt}}=5.
\]

We will first show the solution presented above is feasible for $\mathbf{GP}$. It is straightforward to verify that the solution satisfies the equality constraints of $\mathbf{GP}$. We will therefore focus on the inequality constraints in \eqref{eq:multilp}.

At time $2$, the continuation value is zero. For $c_2=0$ and
$m_2^1=1$, the corresponding $\eta_2^{\mathrm{alt}}(\cdot, 1 \mid 0) = 0$, hence \eqref{eq:multilp} is trivially satisfied. For $c_2=0$ and $m_2^1=0$, the only nontrivial constraint corresponds to $u=1$, and \eqref{eq:multilp} becomes
\[
\sum_x \eta_2^{\mathrm{alt}}(x,0\mid0)r_2^1(x,0) \geq
\sum_x \eta_2^{\mathrm{alt}}(x,0\mid0)r_2^1(x,1),
\]
which gives $0\geq -5$.
For $c_2=1$ and $m_2^1=0$,
\[
\frac12 r_2^1(L,0) \geq \frac12 r_2^1(L,1),
\]
which gives $0\geq -5/2$.
For $c_2=1$ and $m_2^1=1$,
\[
\frac12 r_2^1(H,1) \geq \frac12 r_2^1(H,0),
\]
which gives $5\geq0$.
Thus, \eqref{eq:multilp} holds at time $2$.

At time $1$, first consider $m_1^1=0$. Since
\[
\eta_1^{\mathrm{alt}}(x_1,0,L)
=
\eta_1^{\mathrm{alt}}(x_1,0,H)
=0,
\]
both sides of \eqref{eq:multilp} are zero.

For $m_1^1=1$, the only alternative action is $u=0$. Under the
prescribed action, $C_2=1$, whereas under $u=0$, $C_2=0$.
Therefore,
\eqref{eq:multilp} becomes
\begin{align}
&\sum_{n\in\{L,H\}} \eta_1^{\mathrm{alt}}(x_1,1,n) \big[r_1^1(x_1,1)+W_2^{1,\mathrm{alt}}(1)\big]
\geq
\sum_{n\in\{L,H\}} \eta_1^{\mathrm{alt}}(x_1,1,n)
\big[r_1^1(x_1,0)+W_2^{1,\mathrm{alt}}(0)\big],
\end{align}
which gives
\[
\frac12(0+5)+\frac12(0+5) \geq \frac12(4+0)+\frac12(4+0).
\]
Hence $(g^{d,\mathrm{alt}},\eta^{\mathrm{alt}},
W^{1,\mathrm{alt}},V^{\mathrm{alt}})$ is feasible for $\mathbf{GP}$.

Since $C_1=\emptyset$, the objective value of this feasible solution is $V_1^{\mathrm{alt}}=110$. Therefore, we observe that the alternative solution outperforms the solution returned by Algorithm~\ref{alg:multi agent sequential}. Hence, the solution returned by Algorithm~\ref{alg:multi agent sequential} need
not be globally optimal for $\mathbf{P}(1,r_{1:T}^0)$.

To understand why Algorithm~\ref{alg:multi agent sequential} fails in this case, consider $t=2, c_2=1$: For
Algorithm~\ref{alg:multi agent sequential} we have
$(V_2^{\mathrm{alg}}(1),W_2^{1,\mathrm{alg}}(1))=(20,5/2)$, whereas for the alternative solution we have $(V_2^{\mathrm{alt}}(1),W_2^{1,\mathrm{alt}}(1))=(10,5)$. Even though the alternative solution has a lower designer value function at $t=2$ than the algorithm solution, the corresponding agent value functions change the incentive constraints at $t=1$ in a manner that yields a higher designer value function at $t=1$.

\section{Proof of Theorem \ref{thm:multi agent special cases}}\label{app:special cases}    
Fix $\gamma = 0$.
Let $(g^{d},\allowbreak \eta_{1:T},\allowbreak V_{1:T},\allowbreak W_{1:T}^{1},\allowbreak W_{1:T}^{2})$ be obtained using the sequence of linear programs in Algorithm \ref{alg:multi agent sequential}. That is, for each $t$ and each $c_t$, $(\eta_t(\cdot\mid c_t), g^{d}_t(\cdot\mid \cdot,c_t),\allowbreak W_t^{1}(c_t),\allowbreak W_t^{2}(c_t),\allowbreak V_t(c_t))$ is an optimal solution for $\mathbf{ {OP}_t}(c_t)$. 

Consider an arbitrary feasible strategy $g^{d,global}$ for Problem~\ref{prob:myopic agents}. By Theorem~\ref{thm:GP equivalence}, there exist $\eta_{1:T}^{global},\allowbreak
V_{1:T}^{global},\allowbreak W_{1:T}^{1,global},\allowbreak W_{1:T}^{2,global}$ such that $(g^{d,global},\allowbreak\eta_{1:T}^{global},\allowbreak V_{1:T}^{global},\allowbreak W_{1:T}^{1,global},\allowbreak W_{1:T}^{2,global})$ is feasible for $\mathbf{GP}$. We prove by backward induction that
\begin{equation}
    V_t^{global}(c_t)\leq V_t(c_t),\label{eq:app-myopic-induction}
\end{equation}
for every $t$ and $c_t\in\mathcal C_t$.

\emph{Base case  $(t = T)$:} Fix a $c_T \in \mathcal{C}_T$. Since $W_{T+1}^1(\cdot) = W_{T+1}^2(\cdot) = V_{T+1}(\cdot) = 0$, the restriction of the arbitrary feasible solution to time $T$ and $c_T$, $\bigl(\eta_T^{global}(\cdot\mid c_T), g_T^{d,global}(\cdot\mid\cdot,c_T), W_T^{1,global}(c_T), W_T^{2,global}(c_T), V_T^{global}(c_T)\bigr),$ is feasible for $\mathbf{{OP}_T}(c_T)$. Since $V_T(c_T)$ is the optimal value of $\mathbf{OP}_T(c_T)$,
\begin{equation}
    V_T^{global}(c_T)\leq V_T(c_T).
\end{equation}

\emph{Induction step:} Suppose that  $V^{global}_t(\cdot) \leq V_{t}(\cdot)$ holds for all $t \ge l+1$.  Fix a $c_l \in \mathcal{C}_l$, we obtain
\begin{align}
    V^{global}_l(c_l) &:=           \mathbb{E}^{(g^{d,global},h^1,h^2)_{l:T}}\left[\sum_{k=l}^T r_k^0(X_k, U_k^{0:2}) \Big| C_l = c_l\right] \notag\\
    &= \mathbb{E}^{(g^{d,global},h^1,h^2)_{l:T}}\Big[r_l^0(X_l, U_l^{0:2})+ V^{global}_{l+1}(C_{l+1}) \Big| C_l = c_l\Big] \notag \\
    &=\sum_{\substack{x\in\mathcal{X}_l, p^0\in\mathcal{P}_l^0 \\p^1\in\mathcal{P}_l^1, p^2\in\mathcal{P}_l^2,  m^1 \in \mathcal{M}_l^1 \\ m^2 \in \mathcal{M}_l^2, u^0 \in \mathcal{U}_l^0, n\in \mathcal{N}_l }} \eta^{global}_l(x,p^{0:2}, m^{1,2}, u^0, n | c_l)\times[r_l^0(x,u^0, u_l^1, u_l^2) + V_{l+1}^{global}(c_l, z_{l+1})], \label{eq:app3.a}
            \end{align}
      where $u_l^1 = h_l^1(m^1, p^1, c_l), u_l^2 = h_l^2(m^2, p^2, c_l), z_{l+1} = \zeta_{l+1}(x,\allowbreak p^{0:2}, u^0, u_l^{1,2}, m^{1,2}, n)$. By the induction hypothesis, $V^{global}_{l+1}(\cdot) \le V_{l+1}(\cdot)$. Hence  the expression in \eqref{eq:app3.a} is upper bounded as follows    
\begin{align}
 \eqref{eq:app3.a}   &\le \sum_{\substack{x\in\mathcal{X}_l, p^0\in\mathcal{P}_l^0 \\p^1\in\mathcal{P}_l^1, p^2\in\mathcal{P}_l^2,  m^1 \in \mathcal{M}_l^1 \\ m^2 \in \mathcal{M}_l^2, u^0 \in \mathcal{U}_l^0, n\in \mathcal{N}_l }}
\eta_l^{global}(x,p^{0:2}, m^{1,2}, u^0, n | c_l)\times[r_l^0(x,u^0, u_l^1, u_l^2)+ V_{l+1}(c_l, z_{l+1})]  \label{eq:app3eq1} \\
&:= \hat{V}_l(c_l),
\end{align}
where $u_l^1 = h_l^1(m^1, p^1, c_l), u_l^2 = h_l^2(m^2, p^2, c_l)$ and  $z_{l+1} = \zeta_{l+1}(x, p^{0:2}, u^0, u_l^{1,2}, m^{1,2}, n)$. 

Since $\gamma=0$, the constraints in \eqref{eq:multilp} and the agent
value equations in \eqref{eq:multi agent value function time any t}
do not depend on the future agent value functions. They are therefore
satisfied by $\eta_l^{global}(\cdot\mid c_l)$ and $W_l^{i,global}(c_l)$ for $i\in\{1,2\}$ regardless of the choice of future agent value functions. Finally, \eqref{eq:designer value function time any t multi} is satisfied by $\widehat V_l(c_l)$ above.
Hence, $(\eta_l^{global}(\cdot|c_l), \allowbreak g^{d,global}_l(\cdot|\cdot,c_l),\allowbreak W^{1,global}_l(c_l),\allowbreak W^{2,global}_l(c_l),\allowbreak \hat{V}_l(c_l))$ is feasible for $\mathbf{ {OP}_l}(c_l)$. Since $V_l(c_l)$ comes from the optimal solution for $\mathbf{ {OP}_l}(c_l)$, it follows that 
\begin{equation}\label{eq:app3eq1a}
    \hat{V}_l(c_l) \le  V_l(c_l).
\end{equation}
Combining \eqref{eq:app3eq1} and  \eqref{eq:app3eq1a} gives 
\begin{equation}
    V^{global}_l(c_l) \le \hat{V}_l(c_l) \le V_l(c_l),
\end{equation} 
which completes the induction.

Therefore,
\[
V_1^{global}(c_1)\leq V_1(c_1) \qquad \text{for every }c_1\in\mathcal C_1.
\]
Using \eqref{eq:designer reward reformulation multi},
\begin{align}
J^0(g^{d,global},h^1,h^2)
&=
\sum_{c_1\in\mathcal C_1} \mathbb{P}(C_1=c_1)V_1^{global}(c_1)
\leq
\sum_{c_1\in\mathcal C_1} \mathbb{P}(C_1=c_1)V_1(c_1)
=
J^0(g^d,h^1,h^2).
\end{align}
Since $g^{d,global}$ was an arbitrary feasible strategy for
Problem~\ref{prob:myopic agents}, the strategy $g^d$ returned by
Algorithm~\ref{alg:multi agent sequential} is globally optimal.
\bibliography{ref_journal_2025.bib}

@article{bergemann2019information,
  title={Information design: A unified perspective},
  author={Bergemann, Dirk and Morris, Stephen},
  journal={Journal of Economic Literature},
  volume={57},
  number={1},
  pages={44--95},
  year={2019},
  publisher={American Economic Association 2014 Broadway, Suite 305, Nashville, TN 37203-2425}
}

@article{alonso2016bayesian,
  title={Bayesian persuasion with heterogeneous priors},
  author={Alonso, Ricardo and C{\^a}mara, Odilon},
  journal={Journal of Economic Theory},
  volume={165},
  pages={672--706},
  year={2016},
  publisher={Elsevier}
}

@article{bergemann2016information,
  title={Information design, {B}ayesian persuasion, and Bayes correlated equilibrium},
  author={Bergemann, Dirk and Morris, Stephen},
  journal={American Economic Review},
  volume={106},
  number={5},
  pages={586--591},
  year={2016},
  publisher={American Economic Association 2014 Broadway, Suite 305, Nashville, TN 37203}
}

@article{bergemann2016bayes,
  title={Bayes correlated equilibrium and the comparison of information structures in games},
  author={Bergemann, Dirk and Morris, Stephen},
  journal={Theoretical Economics},
  volume={11},
  number={2},
  pages={487--522},
  year={2016},
  publisher={Wiley Online Library}
}

@article{gentzkow2017bayesian,
  title={Bayesian persuasion with multiple senders and rich signal spaces},
  author={Gentzkow, Matthew and Kamenica, Emir},
  journal={Games and Economic Behavior},
  volume={104},
  pages={411--429},
  year={2017},
  publisher={Elsevier}
}

@techreport{tamura2012theory,
  title={A theory of multidimensional information disclosure},
  author={Tamura, Wataru},
  year={2012},
  institution={ISER Discussion Paper}
}

@article{li2018multplesenders,
author = {Li, Fei and Norman, Peter},
year = {2018},
month = {06},
pages = {},
title = {On {B}ayesian persuasion with multiple senders},
volume = {170},
journal = {Economics Letters},
doi = {10.1016/j.econlet.2018.05.023}
}

@article{akyol2016information,
  title={Information-theoretic approach to strategic communication as a hierarchical game},
  author={Akyol, Emrah and Langbort, C{\'e}dric and Ba{\c{s}}ar, Tamer},
  journal={Proceedings of the IEEE},
  volume={105},
  number={2},
  pages={205--218},
  year={2016},
  publisher={IEEE}
}

@inproceedings{sayin2016strategic,
  title={Strategic control of a tracking system},
  author={Sayin, Muhammed O and Akyol, Emrah and Ba{\c{s}}ar, Tamer},
  booktitle={2016 IEEE 55th Conference on Decision and Control (CDC)},
  pages={6147--6153},
  year={2016},
  organization={IEEE}
}

@article{farokhi2016estimation,
  title={Estimation with strategic sensors},
  author={Farokhi, Farhad and Teixeira, Andr{\'e} MH and Langbort, C{\'e}dric},
  journal={IEEE Transactions on Automatic Control},
  volume={62},
  number={2},
  pages={724--739},
  year={2016},
  publisher={IEEE}
}

@article{ely2017beeps,
  title={Beeps},
  author={Ely, Jeffrey C},
  journal={American Economic Review},
  volume={107},
  number={1},
  pages={31--53},
  year={2017},
  publisher={American Economic Association 2014 Broadway, Suite 305, Nashville, TN 37203}
}

@article{renault2017optimal,
  title={Optimal dynamic information provision},
  author={Renault, J{\'e}r{\^o}me and Solan, Eilon and Vieille, Nicolas},
  journal={Games and Economic Behavior},
  volume={104},
  pages={329--349},
  year={2017},
  publisher={Elsevier}
}

@incollection{sayin2021deception,
  title={Deception-as-defense framework for cyber-physical systems},
  author={Sayin, Muhammed O and Ba{\c{s}}ar, Tamer},
  booktitle={Safety, Security and Privacy for Cyber-Physical Systems},
  pages={287--317},
  year={2021},
  publisher={Springer}
}

@article{kamenica2011bayesian,
  title={Bayesian persuasion},
  author={Kamenica, Emir and Gentzkow, Matthew},
  journal={American Economic Review},
  volume={101},
  number={6},
  pages={2590--2615},
  year={2011},
  publisher={American Economic Association}
}

@book{osborne1994course,
	author = {Osborne, Martin J and Rubinstein, Ariel},
	publisher = {MIT press},
	title = {A course in game theory},
	year = {1994}}

@article{kamenica2019bayesian,
  title={Bayesian persuasion and information design},
  author={Kamenica, Emir},
  journal={Annual Review of Economics},
  volume={11},
  pages={249--272},
  year={2019},
  publisher={Annual Reviews}
}

@article{au2015dynamic,
  title={Dynamic information disclosure},
  author={Au, Pak Hung},
  journal={The RAND Journal of Economics},
  volume={46},
  number={4},
  pages={791--823},
  year={2015},
  publisher={Wiley Online Library}
}

@article{doval2020sequential,
  title={Sequential information design},
  author={Doval, Laura and Ely, Jeffrey C},
  journal={Econometrica},
  volume={88},
  number={6},
  pages={2575--2608},
  year={2020},
  publisher={Wiley Online Library}
}

@article{ely2020moving,
  title={Moving the goalposts},
  author={Ely, Jeffrey C and Szydlowski, Martin},
  journal={Journal of Political Economy},
  volume={128},
  number={2},
  pages={468--506},
  year={2020},
  publisher={The University of Chicago Press Chicago, IL}
}

@inproceedings{farhadi2018static,
  title={Static and dynamic informational incentive mechanisms for security enhancement},
  author={Farhadi, Farzaneh and Teneketzis, Demosthenis and Golestani, S Jamaloddin},
  booktitle={2018 European control conference (ECC)},
  pages={1048--1055},
  year={2018},
  organization={IEEE}
}

@article{meigs2020optimal,
  title={Optimal dynamic information provision in traffic routing},
  author={Meigs, Emily and Parise, Francesca and Ozdaglar, Asuman and Acemoglu, Daron},
  journal={arXiv preprint arXiv:2001.03232},
  year={2020}
}

@inproceedings{tavafoghi2017informational,
  title={Informational incentives for congestion games},
  author={Tavafoghi, Hamidreza and Teneketzis, Demosthenis},
  booktitle={2017 55th Annual Allerton Conference on Communication, Control, and Computing (Allerton)},
  pages={1285--1292},
  year={2017},
  organization={IEEE}
}

@article{farhadi2022dynamic,
  title={Dynamic information design: A simple problem on optimal sequential information disclosure},
  author={Farhadi, Farzaneh and Teneketzis, Demosthenis},
  journal={Dynamic Games and Applications},
  volume={12},
  number={2},
  pages={443--484},
  year={2022},
  publisher={Springer}
}

@article{best2024persuasion,
  title={Persuasion for the long run},
  author={Best, James and Quigley, Daniel},
  journal={Journal of Political Economy},
  volume={132},
  number={5},
  pages={1740--1791},
  year={2024},
  publisher={The University of Chicago Press Chicago, IL}
}

@InProceedings{sayin2019optimality,
author="Sayin, Muhammed O.
and Ba{\c{s}}ar, Tamer",
editor="Alpcan, Tansu
and Vorobeychik, Yevgeniy
and Baras, John S.
and D{\'a}n, Gy{\"o}rgy",
title="On the Optimality of Linear Signaling to Deceive {K}alman Filters over Finite/Infinite Horizons",
booktitle="Decision and Game Theory for Security",
year="2019",
publisher="Springer International Publishing",
address="Cham",
pages="459--478",
isbn="978-3-030-32430-8"
}

@article{sayin2019hierarchical,
  title={Hierarchical multistage {G}aussian signaling games in noncooperative communication and control systems},
  author={Sayin, Muhammed O and Akyol, Emrah and Ba{\c{s}}ar, Tamer},
  journal={Automatica},
  volume={107},
  pages={9--20},
  year={2019},
  publisher={Elsevier}
}

@misc{basar1999dynamic,
  title={Dynamic Noncooperative Game Theory (SIAM)},
  author={Basar, Tamer and Olsder, Geert Jan},
  year={1999},
  publisher={Philadelphia}
}

@book{filar2012competitive,
  title={Competitive Markov decision processes},
  author={Filar, Jerzy and Vrieze, Koos},
  year={2012},
  publisher={Springer Science \& Business Media}
}

@article{maskin2001markov,
  title={Markov perfect equilibrium: I. Observable actions},
  author={Maskin, Eric and Tirole, Jean},
  journal={Journal of Economic Theory},
  volume={100},
  number={2},
  pages={191--219},
  year={2001},
  publisher={Elsevier}
}

@inproceedings{zheng2013decomposition,
  title={Decomposition techniques for {M}arkov zero-sum games with nested information},
  author={Zheng, Jiefu and Casta{\~n}{\'o}n, David A},
  booktitle={52nd IEEE conference on decision and control},
  pages={574--581},
  year={2013},
  organization={IEEE}
}

@article{gensbittel2015value,
  title={The value of {M}arkov chain games with incomplete information on both sides},
  author={Gensbittel, Fabien and Renault, J{\'e}r{\^o}me},
  journal={Mathematics of Operations Research},
  volume={40},
  number={4},
  pages={820--841},
  year={2015},
  publisher={INFORMS}
}

@inproceedings{li2014lp,
  title={LP formulation of asymmetric zero-sum stochastic games},
  author={Li, Lichun and Shamma, Jeff},
  booktitle={53rd IEEE conference on decision and control},
  pages={1930--1935},
  year={2014},
  organization={IEEE}
}

@article{kartik2021upper,
  title={Upper and lower values in zero-sum stochastic games with asymmetric information},
  author={Kartik, Dhruva and Nayyar, Ashutosh},
  journal={Dynamic Games and Applications},
  volume={11},
  pages={363--388},
  year={2021},
  publisher={Springer}
}

@article{gupta2014common,
  title={Common information based {M}arkov perfect equilibria for linear-Gaussian games with asymmetric information},
  author={Gupta, Abhishek and Nayyar, Ashutosh and Langbort, C{\'e}dric and Basar, Tamer},
  journal={SIAM Journal on Control and Optimization},
  volume={52},
  number={5},
  pages={3228--3260},
  year={2014},
  publisher={SIAM}
}

@article{gupta2016dynamic,
  title={Dynamic games with asymmetric information and resource constrained players with applications to security of cyberphysical systems},
  author={Gupta, Abhishek and Langbort, C{\'e}dric and Ba{\c{s}}ar, Tamer},
  journal={IEEE Transactions on Control of Network Systems},
  volume={4},
  number={1},
  pages={71--81},
  year={2016},
  publisher={IEEE}
}

@inproceedings{ouyang2015dynamic,
  title={Dynamic oligopoly games with private Markovian dynamics},
  author={Ouyang, Yi and Tavafoghi, Hamidreza and Teneketzis, Demosthenis},
  booktitle={2015 54th IEEE Conference on Decision and Control (CDC)},
  pages={5851--5858},
  year={2015},
  organization={IEEE}
}

@article{nayyar2013common,
  title={Common information based Markov perfect equilibria for stochastic games with asymmetric information: Finite games},
  author={Nayyar, Ashutosh and Gupta, Abhishek and Langbort, Cedric and Ba{\c{s}}ar, Tamer},
  journal={IEEE Transactions on Automatic Control},
  volume={59},
  number={3},
  pages={555--570},
  year={2013},
  publisher={IEEE}
}

@article{ouyang2016dynamic,
  title={Dynamic games with asymmetric information: Common information based perfect bayesian equilibria and sequential decomposition},
  author={Ouyang, Yi and Tavafoghi, Hamidreza and Teneketzis, Demosthenis},
  journal={IEEE Transactions on Automatic Control},
  volume={62},
  number={1},
  pages={222--237},
  year={2016},
  publisher={IEEE}
}

@INPROCEEDINGS{tavafoghi2016stochastic,
author={Tavafoghi, Hamidreza and Ouyang, Yi and Teneketzis, Demosthenis},
booktitle={2016 IEEE 55th Conference on Decision and Control (CDC)}, 
title={On stochastic dynamic games with delayed sharing information structure}, 
year={2016},
volume={},
number={},
pages={7002-7009},
doi={10.1109/CDC.2016.7799348}}

@article{vasal2018systematic,
  title={A systematic process for evaluating structured perfect Bayesian equilibria in dynamic games with asymmetric information},
  author={Vasal, Deepanshu and Sinha, Abhinav and Anastasopoulos, Achilleas},
  journal={IEEE Transactions on Automatic Control},
  volume={64},
  number={1},
  pages={81--96},
  year={2018},
  publisher={IEEE}
}

@article{tang2023dynamic,
  title={Dynamic games among teams with delayed intra-team information sharing},
  author={Tang, Dengwang and Tavafoghi, Hamidreza and Subramanian, Vijay and Nayyar, Ashutosh and Teneketzis, Demosthenis},
  journal={Dynamic Games and Applications},
  volume={13},
  number={1},
  pages={353--411},
  year={2023},
  publisher={Springer}
}

@article{li2018lp,
  title={An LP approach for solving two-player zero-sum repeated Bayesian games},
  author={Li, Lichun and Langbort, Cedric and Shamma, Jeff},
  journal={IEEE Transactions on Automatic Control},
  volume={64},
  number={9},
  pages={3716--3731},
  year={2018},
  publisher={IEEE}
}

@InProceedings{pmlr-v202-liu23ay,
  title = 	 {Partially Observable Multi-agent {RL} with ({Q}uasi-){E}fficiency: The Blessing of Information Sharing},
  author =       {Liu, Xiangyu and Zhang, Kaiqing},
  booktitle = 	 {Proceedings of the 40th International Conference on Machine Learning},
  pages = 	 {22370--22419},
  year = 	 {2023},
  editor = 	 {Krause, Andreas and Brunskill, Emma and Cho, Kyunghyun and Engelhardt, Barbara and Sabato, Sivan and Scarlett, Jonathan},
  volume = 	 {202},
  series = 	 {Proceedings of Machine Learning Research},
  month = 	 {23--29 Jul},
  publisher =    {PMLR},
  url = 	 {https://proceedings.mlr.press/v202/liu23ay.html}
}


\end{document}